\documentclass[11pt,a4paper,reqno,oneside]{amsart}
\usepackage[utf8]{inputenc}
\usepackage[T1]{fontenc}
\usepackage[english]{babel}
\usepackage[a4paper]{geometry}
\usepackage{lmodern,textcomp,amsmath,amssymb,mathtools,braket,authblk,amsthm,bm,comment}
\usepackage[foot]{amsaddr}
\usepackage[shortlabels]{enumitem}

\usepackage[unicode=true,bookmarks=true,bookmarksopen=false,breaklinks=false,pdfborder={0 0 0},colorlinks=true]{hyperref}
\usepackage{xcolor}
\definecolor{cblue}{rgb}{0.16, 0.32, 0.75}
\definecolor{cred}{rgb}{0.7, 0.11, 0.11}
\hypersetup{%
	,linkcolor=cred
	,citecolor=cblue
	,urlcolor=black
}

\makeatletter
\def\@setemails{%
 \mbox{{\itshape ${}^*$Corresponding author}:\space}{\ttfamily\emails}.%
}
\makeatother

\newcommand{\hilb}{\mathcal{H}}
\newcommand{\e}{\mathrm{e}}
\newcommand{\ii}{\mathrm{i}}

\numberwithin{equation}{section}

\newtheorem{theorem}{Theorem}[section]
\newtheorem{corollary}[theorem]{Corollary}
\newtheorem{lemma}[theorem]{Lemma}
\newtheorem{proposition}[theorem]{Proposition}
\theoremstyle{definition}
\newtheorem{definition}[theorem]{Definition}
\newtheorem{hypothesis}[theorem]{Hypothesis}
\theoremstyle{remark}
\newtheorem{remark}[theorem]{Remark}
\newtheorem{example}[theorem]{Example}
\AtBeginEnvironment{example}{%
	\pushQED{\qed}%
}
\AtEndEnvironment{example}{\popQED\endexample}
\AtBeginEnvironment{remark}{%
	\pushQED{\qed}%
}
\AtEndEnvironment{remark}{\popQED\endexample}

\usepackage{orcidlink}

\newcommand\blfootnote[1]{%
  \begingroup
  \renewcommand\thefootnote{}\footnote{#1}%
  \addtocounter{footnote}{-1}%
  \endgroup
}

\DeclareMathOperator{\Dom}{Dom}

\renewcommand{\d}{\mathrm{d}}
\newcommand{\fock}{\mathcal{F}}
\newcommand{\hfrak}{\mathfrak{H}}
\newcommand{\sps}{\mathfrak{h}}

\renewcommand{\a}[1]{a\!\left(#1\right)}
\newcommand{\adag}[1]{a^\dag\!\left(#1\right)}
\newcommand{\dOmega}{\mathrm{d}\Gamma(\omega)}

\usepackage[
backend=bibtex,
style=numeric-comp,
giveninits=true,
maxbibnames=299,
natbib=true,
doi=false,
isbn=false,
url=false,
sorting=nyt
]
{biblatex}
\renewbibmacro{in:}{}
\DeclareFieldFormat[misc]{title}{``#1''\isdot}
\DeclareFieldFormat{journaltitle}{#1\isdot}
\DeclareFieldFormat[article,periodical]{volume}{\mkbibbold{#1}}
\DeclareFieldFormat{pages}{#1}
\AtEveryBibitem{%
	\clearfield{number}}

\allowdisplaybreaks

\title{Rotating-wave approximation for spin--boson models with structured fields}

\author{Aitor Balmaseda\textsuperscript{\,1}\hspace{2pt}\orcidlink{0000-0001-9624-9692}\hspace{1pt}}
\author{Davide Lonigro\textsuperscript{\,2}\hspace{2pt}\orcidlink{0000-0002-0792-8122}\hspace{1pt}}
\author{Juan Manuel P\'erez-Pardo\textsuperscript{\,1}\hspace{2pt}\orcidlink{0000-0003-2424-8943}\hspace{1pt}}
\address{\footnotesize \textsuperscript{1}Departamento de Matem\'aticas, Universidad Carlos III de Madrid, Avda.\ de la Universidad 30, 28911 Madrid, Spain}
\address{\footnotesize \textsuperscript{2}Department of Physics, Friedrich-Alexander-Universit\"at Erlangen-N\"urnberg, Staudtstra\ss e 7, 91058 Erlangen, Germany}

\begin{document}

\maketitle
\thispagestyle{empty}

\begin{abstract}
We derive state-dependent bounds on the difference between two quantum evolutions generated by unbounded Hamiltonians sharing a common form domain. The main technical tool is a second integration by parts, performed at the level of sesquilinear forms rather than at the operator level, which removes the need for a common invariant operator domain. The resulting estimate involves the norm of the time-integrated difference of the two generators, rather than the integral of its norm, and is therefore sensitive to the averaging effects produced by fast-oscillating terms.

As an application we prove a quantitative bound on the rotating-wave approximation for spin--boson models with a structured boson field, described by an arbitrary massive dispersion relation on a general measure space and by a suitable class of form factors. The proof involves a careful analysis of the high-frequency scaling. The bound holds on a dense subspace of states, is fully explicit, and all the constants entering it depend only on the parameters of the model and not on the frequency scale, so that the approximation becomes exact in the limit of large frequency.
\end{abstract}

\noindent
\vspace{0.25cm}


\section{Introduction}

The evolution of closed quantum systems is described by the unitary group generated by a self-adjoint operator---the Hamiltonian. In the general case these operators are typically unbounded and may explicitly depend on time. In most relevant situations, however, the corresponding unitary dynamics cannot be computed exactly, and one must resort to approximations. 
Having bounds that provide quantitative estimates of the deviation between the true and the approximated dynamics is therefore of major importance. Such bounds make it possible to determine the range of validity of the approximation in terms of the parameters of the problem. For instance, one cannot expect an approximation to be valid over arbitrarily large periods of time or different approximations might be incompatible with each other if the parameters of the problem are not chosen carefully, e.g. \cite{augier2019compatibility}.

The aim of this article is to address this problem rigorously and in a general setting. We will compare the evolutions generated by two (possibly time-dependent) Hamiltonians. 
More precisely, let $\{H_1(t)\}_t$ and $\{H_2(t)\}_t$, $t\in[0,T]$, be two families of densely defined and self-adjoint operators with domains $\Dom H_j(t)$, $j=1,2$. That is, for each $t\in[0,T]$ the operator $H_j(t)$ with domain $\Dom H_j(t)$ is a self-adjoint operator on the Hilbert space $\hilb$. 
These families of operators, called time-dependent Hamiltonians, appear often when modelling external interactions in quantum mechanical systems. This is the case, for instance, when the strength of the coupling between the interaction and the free dynamics varies with time or when considering convenient time-dependent transformations, like the so-called interaction picture, which is the basis for the \emph{rotating-wave approximation} (RWA).
We assume that both Hamiltonians generate unique unitary propagators $U_j(t,s)$, providing the unique solutions to the Schr\"odinger equation
\begin{equation}
    \ii\frac{\mathrm{d}}{\mathrm{d}t}U_j(t,s)\Psi=H_j(t)U_j(t,s)\Psi,\qquad \Psi\in\Dom H_j(s).
\end{equation}
 
To illustrate the main ideas of the proof of the main mathematical result, Theorem~\ref{thm:bound}, let us briefly consider the case where, for each $t$, $H_j(t)$ is a bounded operator and $\|H_j(t)\|<B$ with $B$ independent of $t$. In this case the Schr\"odinger equation holds in the norm topology:
\begin{equation}
    \ii\frac{\mathrm{d}}{\mathrm{d}t}U_j(t,s)=H_j(t)U_j(t,s),
\end{equation}
and all operators are defined on the whole Hilbert space $\hilb$. A first simple bound follows from the identities
\begin{align}\label{eq:identities}
    U_2(t,s)-U_1(t,s)&=U_1(t,s)\left[U_1^\dag(t,s)U_2(t,s)-1\right]\nonumber\\
    &=U_1(t,s)\int_s^t\frac{\mathrm{d}}{\mathrm{d}\tau}U_1^\dag(\tau,s)U_2(\tau,s)\;\mathrm{d}\tau\nonumber\\
    &=-\ii\, U_1(t,s)\int_s^tU_1^\dag(\tau,s)\left[H_2(\tau)-H_1(\tau)\right]U_2(\tau,s)\;\mathrm{d}\tau,
\end{align}
where in the last step we used the Schr\"odinger equation for both propagators. Taking norms and applying the triangle inequality yields
\begin{equation}\label{eq:bound1}
    \|U_2(t,s)-U_1(t,s)\|\leq\int_s^t\|H_2(\tau)-H_1(\tau)\|\;\mathrm{d}\tau.
\end{equation}
Thus the norm of the difference between the two evolutions is controlled by the time-integrated norm of the difference between the two Hamiltonians. 
Despite being loose at large times (the left-hand side is at most $2$), this bound correctly reflects the intuition that similar Hamiltonians yield similar evolutions. The converse, however, need not hold: Hamiltonians whose difference has large norm can still generate similar evolutions if those oscillate rapidly. In other words, this bound does not capture \textit{averaging effects}. This is, for instance, the case of the RWA in which the terms neglected by the approximation are fast oscillating terms. The article \cite{burgarth2022one} shows that a second integration by parts in Eq.~\eqref{eq:identities} can provide upper estimates that also bound this kind of fast oscillating terms. This yields
\begin{align}
    U_2(t,s)-U_1(t,s)=&-\ii S_{21}(t,s)U_2(t,s)\nonumber\\&-U_1(t,s)\int_s^tU_1^\dag(\tau,s)\left[H_1(\tau)S_{21}(\tau,s)-S_{21}(\tau,s)H_2(\tau)\right]U_2(\tau,s)\;\mathrm{d}\tau,
\end{align}
where $S_{21}(t,s)$ is the integrated difference of the two Hamiltonians:
\begin{equation}
    S_{21}(t,s)=\int_s^t\left[H_2(\tau)-H_1(\tau)\right]\;\mathrm{d}\tau;
\end{equation}
whence
\begin{equation}\label{eq:bound2}
    \|U_2(t,s)-U_1(t,s)\|\leq\|S_{21}(t,s)\|+\int_s^t\|S_{21}(\tau,s)\|\left(\|H_1(\tau)\|+\|H_2(\tau)\|\right)\;\mathrm{d}\tau.
\end{equation}
Here the error is controlled by the \textit{norm of the integral} of $H_2(\tau)-H_1(\tau)$ rather than the integral of its norm. This distinction is crucial when fast-oscillating terms are present, since the former can be small even when the latter is not. This technique was later generalized in \cite{dey2025floquet} to incorporate higher-order averaging effects, and was shown to reproduce the Floquet--Magnus expansion whenever the latter is well defined.

Both bounds above are formulated in the operator norm topology; as such, they have two immediate shortcomings. First, they are not applicable to unbounded Hamiltonians, excluding many physically relevant models. Second, even in the bounded case, they provide only a \textit{worst-case} bound over all initial states, since
    \begin{equation}
        \|U_2(t,s)-U_1(t,s)\|=\sup_{\|\Psi\|=1}\|U_2(t,s)\Psi-U_1(t,s)\Psi\|.
    \end{equation}
Overcoming these limitations involves nontrivial mathematical challenges. First, when $H_j(t)$ is unbounded, the existence of a unitary propagator $U_j(t,s)$ is not guaranteed even for smoothly time-dependent Hamiltonians~\cite{kato1953, Kisynski1964, Simon1971, reed1975ii,  balmaseda2021schrodinger}. Second, even when such propagators exist, establishing analogues of the bounds above requires additional domain assumptions, since the operator $H_2(t)-H_1(t)$ may be defined only on a non-dense or even trivial subspace.

\subsection{Results}

We tackle this problem via a form-based approach. The idea of capturing averaging effects via a second integration by parts originates in \cite{burgarth2022one}; it was then applied to derive RWA error bounds at the operator level in \cite{burgarth2024taming,richter2026dicke}, and extended to higher-order effective Hamiltonians for bounded and unbounded operators in \cite{dey2025floquet} and \cite{burgarth2026floquet}, respectively. What these works have in common is that they operate at the operator level, either on the full Hilbert space (bounded case) or on a carefully identified common invariant domain (unbounded case). The latter conditions are restrictive and cannot always be met in concrete applications.

The novelty of the present article lies in circumventing these domain difficulties by carrying the second integration by parts to the level of sesquilinear forms, closed and semibounded instances of which are uniquely associated with self-adjoint operators, cf.\ \cite{kato2013perturbation, Davies1995, reed1981functional}.
The sesquilinear form approach provides a convenient functional-analytical setting that allows to circumvent many technical difficulties that arise at the operator level, e.g., \cite{ibort2020representation, ibort2013numerical, lopez2017finite, ibort2015self}. 

Concretely, this builds on \cite{balmaseda2023sharper}, where a state-dependent bound generalizing Eq.~\eqref{eq:bound1} was derived for Hamiltonians with constant form domain. 
This framework has proven successful in tackling non-autonomous evolution problems and their applications to quantum control \cite{balmaseda2024global, balmaseda2023global, balmaseda2021controllability}.
One considers Hamiltonians $H_1(t),H_2(t)$ associated with sesquilinear forms $h_{1,t},h_{2,t}:\hilb^+\times\hilb^+\rightarrow\mathbb{C}$, where $\hilb^+\subset\hilb$ carries a finer topology with norm $\|\cdot\|_+\geq\|\cdot\|$; each $H_j(t)$ then extends to a bounded map, that we will keep denoting with the same symbol, ${H}_j(t)\in\mathcal{B}(\hilb^+,\hilb^-)$, with $\hilb^-$ the dual of $\hilb^+$. While \cite{balmaseda2023sharper} handles unbounded operators, its bound still involves the \textit{integral of the norm} of ${H}_2(\tau)-{H}_1(\tau)$ and so misses averaging effects. We show that the second integration by parts of \cite{burgarth2022one} can be performed at the form level, yielding a bound involving the \textit{norm of the integral} ${S}_{21}(t,s)=\int_s^t[{H}_2(\tau)-{H}_1(\tau)]\,\mathrm{d}\tau$ in an appropriate functional setting. 

As an application, we derive an RWA error bound for spin--boson models with a general boson field, briefly explained below and properly introduced in Section~\ref{sec:rwa}, where the explicit bounds are given. 

Generic Hamiltonians for the coupling between a spin system and a bosonic field are the following
\begin{eqnarray}
	H&=&H_{\rm free}+\,(\sigma_+ + \sigma_-)\otimes\left(\a{f}+\adag{f}\right),\\
	H_{\rm RWA}&=&H_{\rm free}+\left(\sigma_+\otimes\a{f}+\sigma_-\otimes\adag{f}\right),
\end{eqnarray}
where $\sigma_+,\sigma_- \in \mathbb{C}^{2\times2}$ are the spin raising and lowering matrices, with $\sigma_-=\sigma_+^\dag$, $H_{\rm free}$ is the Hamiltonian governing the non-interacting dynamics and $\a{f},\adag{f}$ are the annihilation and creation operators smeared against $f$. We refer to Section~\ref{sec:prelim}, cf.\ Eqs.~\eqref{eq:spin-boson1}--\eqref{eq:spin-boson5}, for an introduction and to \cite{bratteli1997,reed1975ii} for further details. 

The Hamiltonian $H_{\rm RWA}$ is obtained from $H$ by the \textit{rotating-wave approximation} (RWA): one discards the terms $\sigma_+\otimes\adag{f}$ and $\sigma_-\otimes\a{f}$, which do not conserve the total excitation number, retaining only the two number-conserving terms. This makes the model considerably more tractable. Despite the approximation being nearly as old as quantum optics~\cite{jaynes1963,allen1987}, a rigorous and \textit{quantitative} justification was obtained only recently, in \cite{burgarth2024taming}, for the monochromatic case in which $H,H_{\rm RWA}$ reduce to the so-called Rabi~\cite{rabi1936,rabi1937} and Jaynes--Cummings~\cite{jaynes1963} models.
The main contribution of this article is a theorem that also treats the non-monochromatic setting. We allow for a general dispersion relation $\omega(k)$, in which case the eliminated terms carry a continuum of frequencies. We have obtained an abstract bound, Theorem~\ref{thm:bound}, and have applied it to treat the general case, obtaining the explicit estimate of Theorem~\ref{thm:explicit}; its scope and limitations are discussed in Section~\ref{sec:discussion}.

\subsection{Context and related work}

The RWA has played a central role in quantum optics since the earliest days of the field~\cite{jaynes1963,allen1987}: it is routinely invoked to render otherwise intractable models accessible, the quintessential example being the reduction of the quantum Rabi model~\cite{rabi1936,rabi1937} to the Jaynes--Cummings model~\cite{jaynes1963}. The approximation is standardly justified by heuristic averaging arguments~\cite{allen1987}, and it is well understood at this informal level. A rigorous quantitative theory has emerged only recently.

The key step was taken in \cite{burgarth2024taming}, where the second integration-by-parts technique described above was applied, at the operator level and on a common invariant domain, to derive explicit error bounds for the RWA of the quantum Rabi model~\cite{braak2011} (monochromatic field), yielding bounds of order $O(\alpha^{-1})$ in the large-frequency parameter $\alpha$. Those bounds are themselves state-dependent, the Hamiltonians involved being unbounded. Related techniques were developed in \cite{burgarth2022one} to unify adiabatic, Zeno, and averaging limits in a single framework, and in \cite{heib2025rwa} for systems of coupled harmonic oscillators. The approach of \cite{burgarth2024taming} was extended to the multi-spin (Dicke) setting in \cite{richter2026dicke}; while the model involves more degrees of freedom, the field remains monochromatic, so the counter-rotating terms oscillate at a single frequency and the integration-by-parts strategy applies at the operator level on a common invariant domain. 

A parallel line of work has aimed at going beyond first order. In \cite{dey2025floquet} an iterated integration-by-parts technique was used to construct effective Hamiltonians of arbitrary order for periodically driven systems in the bounded setting, recovering the Floquet--Magnus expansion with explicit error bounds at each order. 
The latter was recently extended to unbounded Hamiltonians in \cite{burgarth2026floquet} at the operator level, where higher-order corrections to the RWA were obtained on some carefully chosen domains.

More broadly, quantitative control of effective dynamics for driven quantum systems is a topic of wide interest~\cite{blanes2009,fleming2010rwa}. Dynamical decoupling of spin--boson models, which shares the same mathematical structure (comparing evolutions under two Hamiltonians in the large-frequency limit), was treated rigorously in \cite{hahn2025dd}; that analysis also covers structured, although only discrete, boson fields, and likewise requires regularity conditions on the form factor of the same nature as those imposed here. Spin--boson models with singular form factors were studied in \cite{dammoller,lonigro2022generalized,lonigro2023selfadjoint,Lill2025}; that regime lies outside the scope of the present analysis, cf.\ Hypothesis~\ref{hyp:f}. The theory of form-linear Hamiltonians developed in \cite{balmaseda2023sharper} provides a tool that we combine with a direct energy estimate to derive a fully explicit version of the RWA bound. The present article gives, to our knowledge, the first rigorous quantitative RWA error bound for spin--boson models beyond the monochromatic case.

\subsection{Outline}

The paper is organized as follows. Section~\ref{sec:general} presents the abstract framework: after fixing notation and hypotheses, we derive the main form-level identity (Proposition~\ref{prop:equality}) and the resulting state-dependent bound (Theorem~\ref{thm:bound}). Section~\ref{sec:rwa} applies this to the spin--boson model: we introduce the model, carry out the interaction-picture transformation, compute the relative action explicitly, and in §\ref{sec:rwa_bound_subsec} derive the state-dependent RWA bound Eq.~\eqref{eq:rwa_bound} and a fully explicit bound in terms of the problem parameters, Theorem~\ref{thm:explicit}, in which all constants are given in closed form using direct energy estimates. Section~\ref{sec:discussion} discusses the scope of the result, the limitations of the method---in particular the rate obtained and the regularity required of the form factor---and the questions it leaves open.

\section{State-dependent bounds: definitions and main results}\label{sec:general}

Let $\hilb$ be a Hilbert space with scalar product $\braket{\cdot,\cdot}$ and associated norm $\|\cdot\|$, and $\hilb^+$, with scalar product $\braket{\cdot,\cdot}_+$ and associated norm $\|\cdot\|_+$, a Hilbert space which is densely and continuously embedded in $\hilb$ with $\|\cdot\|\leq\|\cdot\|_+$.
Given an interval $I\subset\mathbb{R}$, we consider for every $t\in I$ two time-dependent sesquilinear forms
\begin{equation}
    h_{j,t}:\hilb^+\times\hilb^+\rightarrow\mathbb{C},\qquad(\Psi,\Phi)\in\hilb^+\times\hilb^+\mapsto h_{j,t}(\Psi,\Phi)\in\mathbb{C},\qquad j=1,2\,,
\end{equation}
satisfying the following assumptions:

\begin{hypothesis}\label{hyp:1}
The forms $h_{j,t}$ are Hermitian, i.e.
\begin{equation}\label{eq:hermitian}
    h_{j,t}(\Phi,\Psi)=\overline{h_{j,t}(\Psi,\Phi)}\qquad\text{for all }\Psi,\Phi\in\hilb^+,\;t\in I,
\end{equation}
for all $\Psi,\Phi\in\hilb^+$ the map $t\mapsto h_{j,t}(\Psi,\Phi)$ is continuous on $I$, and there exist $K_j>0$, $j=1,2$, such that
\begin{equation}\label{eq:form}
    |h_{j,t}(\Psi,\Phi)|\leq K_j\|\Psi\|_+\|\Phi\|_+\qquad\text{for all }\Psi,\Phi\in\hilb^+,\;t\in I.
\end{equation}
\end{hypothesis}
\begin{remark}
By the Riesz representation theorem, Hypothesis~\ref{hyp:1} is equivalent to the following. Let $\hilb^-$ be the dual of $\hilb^+$, with $\hilb^+\subset\hilb\subset\hilb^-$ \cite{schmudgen2012,reed1975ii, Berezanskii1968}. Then for each $j=1,2$ and $t\in I$ there is a unique operator ${H}_j(t)\in\mathcal{B}(\hilb^+,\hilb^-)$ such that
\begin{equation}
    h_{j,t}(\Psi,\Phi)=\left(\Psi,{H}_j(t)\Phi\right)_{+,-},
\end{equation}
where $(\cdot,\cdot)_{+,-}\colon \hilb^+\times\hilb^-\to\mathbb{C}$ denotes the duality pairing between $\hilb^+$ and $\hilb^-$. The optimal constant $K_j$ in~\eqref{eq:form} equals the operator norm $\|{H}_j(t)\|_{+,-}$. Condition~\eqref{eq:hermitian} says that ${H}_j(t)$ is symmetric with respect to the duality pairing.
\end{remark}

\begin{hypothesis}\label{hyp:2}
    There exist two unitary propagators $U_j(t,s)$, $j=1,2$, $t,s\in I$ having the following properties:
    \begin{itemize}
        \item[(i)] $U_j(t,s)\hilb^+\subset\hilb^+$ and
    \begin{equation}\label{eq:weak_schro}
        \ii\frac{\mathrm{d}}{\mathrm{d}t}\braket{\Psi,U_j(t,s)\Phi}=h_{j,t}\left(\Psi,U_j(t,s)\Phi\right)\qquad\text{for all }\Psi,\Phi\in\hilb^+,\;t,s\in I,
    \end{equation}
    or equivalently
    \begin{equation}
    \ii\frac{\mathrm{d}}{\mathrm{d}t}U_j(t,s)={H}_j(t)U_j(t,s)
\end{equation}
holds in $\mathcal{B}(\hilb^+,\hilb^-)$;
\item[(ii)] there exists a subspace $\mathcal{D}\subset\hilb^+$, dense in $\hilb$, such that
\begin{itemize}
    \item[(a)] $U_j(t,s)\mathcal{D}\subset\mathcal{D}$ for all $j=1,2$ and $t,s\in I$;
    \item[(b)] ${H}_j(\tau)\mathcal{D}\subset\hilb^+$ for all $j=1,2$ and $\tau\in I$;
    \item[(c)] for every $\Psi\in\mathcal{D}$ and all $j,k=1,2$ the map
    \begin{equation}
        (\tau,t,s)\in I\times I\times I\longmapsto {H}_j(\tau)U_k(t,s)\Psi\in\hilb^+
    \end{equation}
    is continuous, and there is a constant $M_\Psi>0$, depending on $\Psi$ but not on the remaining parameters, such that
    \begin{equation}\label{eq:uniform_orbit}
        \left\|{H}_j(\tau)U_k(t,s)\Psi\right\|_+\leq M_\Psi\qquad\text{for all }j,k=1,2\text{ and }\tau,t,s\in I.
    \end{equation}
\end{itemize}
    \end{itemize} 
\end{hypothesis}
 In the left-hand side of the second equation in (i), $U_j(t,s)$ is interpreted as a map from $\hilb^+$ to $\hilb^-$, and the derivative is understood in the corresponding topology. Note that {(a)} and {(b)} together already give ${H}_j(\tau)U_k(t,s)\mathcal{D}\subset\hilb^+$, so that the quantity in~\eqref{eq:uniform_orbit} is well defined; the content of {(c)} is that it stays bounded along the orbit of $\Psi$ under both propagators, uniformly in all the time variables.

 \begin{remark}
     Note that we have not required $h_{j,t}$ to be closed or bounded from below, so the existence of a unitary propagator satisfying Hypothesis~\ref{hyp:2} is not automatic. This is intentional: we want to cover situations where $h_{j,t}$ is not bounded from below, yet a suitable transformation guarantees the existence of the propagator.
 \end{remark}

The two hypotheses above are abstract and make no reference to any particular form of the operators. 
Hypothesis~\ref{hyp:2}(i) is a minimal requirement that ensures the well-posedness of the dynamical problem, and is satisfied in wide generality. Hypothesis~\ref{hyp:2}(ii) is of a different nature: it requires the generators to map a subspace of $\hilb^+$ back \emph{into} $\hilb^+$, rather than into $\hilb^-$, and this cannot be inferred from the boundedness of the forms alone (see Remark~\ref{rem:ii_needed} below). Verifying it typically demands higher regularity, both of the perturbation and of its time-dependence, cf.\ \cite{Kisynski1964}. The following examples illustrate how Hypothesis~\ref{hyp:2}(i) arises in practice; Hypothesis~\ref{hyp:2}(ii) is verified for the model of interest in Section~\ref{sec:rwa}, cf.\ Proposition~\ref{prop:uniform_energy_bound}. 

\begin{example}\label{ex:bilinear}
    Let $H_0$ be a nonnegative self-adjoint operator on $\hilb$, with associated nonnegative form $h_0:\hilb^+\times\hilb^+\rightarrow\mathbb{C}$ where $\hilb^+$ is the form domain of $H_0$. Let $v$ be a sesquilinear form relatively bounded with respect to $h_0$, i.e., there exist $a,b\geq0$ such that
    \begin{equation}
        |v(\Phi,\Phi)|\leq a\,h_0(\Phi,\Phi)+b\|\Phi\|^2,
    \end{equation}
    and let $\varphi$ be a continuous, piecewise-$C^1$ real-valued function on a bounded interval $I\subset\mathbb{R}$, with $\kappa:=\max_{t\in I}|\varphi(t)|$. If $\kappa<1/a$, then by the KLMN theorem~\cite{reed1975ii,schmudgen2012} the form $h(t):=h_0+\varphi(t)v$ satisfies Hypothesis~\ref{hyp:1}; moreover, as shown in \cite{balmaseda2023sharper}, it satisfies Hypothesis~\ref{hyp:2}(i) as well. Hamiltonians of this type define the so-called form-bilinear control systems, studied in \cite{balmaseda2024global, balmaseda2023global}.
\end{example}

\begin{example}\label{ex:interaction}
    Let $h_0,v$ be as in the previous example, with $0<g<1/a$. By the KLMN theorem the time-independent form $h:=h_0+gv$ defines a self-adjoint operator $H$ bounded from below, and hence a time-homogeneous propagator $U(t,s):=\e^{-\ii (t-s)H}$. Setting $s=0$ for simplicity, the \textit{interaction picture} propagator is
    \begin{equation}
        U^{\rm I}(t,0):=\e^{\ii tH_0}\e^{-\ii tH}.
    \end{equation}
This propagator is the solution of the non-autonomous Schrödinger equation generated by the time-dependent Hamiltonian $H^{\rm I}(t):=g\e^{\ii tH_0}V\e^{-\ii tH_0}$, where $V$ is the symmetric operator on $\hilb$ associated with the form $v$. Rigorously, $U^{\rm I}(t,s)$ is the solution of the non-autonomous Schrödinger equation defined in  weak form by the form $h^{\rm I}_t:\hilb^+\times\hilb^+\rightarrow\mathbb{C}$ given by
\begin{equation}
    h^{\rm I}_t(\Psi,\Phi):=gv\left(\e^{-\ii tH_0}\Psi,\e^{-\ii tH_0}\Phi\right),
\end{equation}
which is well-defined since $\e^{-\ii tH_0}\hilb^+=\hilb^+$. One checks that
\begin{align}
    |h^{\rm I}_t(\Phi,\Phi)|&\leq ag\,h_0\left(\e^{-\ii tH_0}\Phi,\e^{-\ii tH_0}\Phi\right)+bg\,\|\e^{-\ii tH_0}\Phi\|^2\nonumber\\
    &=ag\,h_0(\Phi,\Phi)+bg\|\Phi\|^2\nonumber\\
    &\leq g\max\{a,b\}\|\Phi\|^2_+,
\end{align}
so Hypothesis~\ref{hyp:1} holds. For Hypothesis~\ref{hyp:2}, note that we cannot appeal to \cite{balmaseda2023sharper} since $v$ is not assumed bounded from below. Instead, the inclusion $U^{\rm I}(t,0)\hilb^+\subset\hilb^+$ follows from
\begin{equation}
    \e^{\ii tH_0}\hilb^+\subset\hilb^+,\qquad\e^{-\ii tH}\hilb^+\subset\hilb^+,
\end{equation}
and the Schr\"odinger equation in $\mathcal{B}(\hilb^+,\hilb^-)$ follows from the chain rule applied to both factors, so that Hypothesis~\ref{hyp:2}(i) holds. The natural candidate for the subspace $\mathcal{D}$ of Hypothesis~\ref{hyp:2}(ii) is
\begin{equation}
    \left\{\Phi\in\hilb^+:{H}^{\rm I}(t)\Phi\in\hilb^+\ \text{ for all }t\in I\right\}=\left\{\Phi\in\hilb^+:\e^{\ii tH_0}V\e^{-\ii tH_0}\Phi\in\hilb^+\ \text{ for all }t\in I\right\},
\end{equation}
but establishing that this set is dense and satisfies {(a)}--{(c)} requires additional information about $V$ that the present assumptions alone do not provide; see Remark~\ref{rem:ii_needed}.
\end{example}

The two examples above illustrate how time-dependence might appear in the generators. In the first example it is explicit in the generator; in the second example it is implicitly encoded in the interaction picture. Intermediate cases can be treated similarly.

\begin{remark}\label{rem:ii_needed}
Neither example automatically verifies Hypothesis~\ref{hyp:2}(ii), and in fact none of the assumptions made so far implies it. Indeed, let $H_0$ be any nonnegative self-adjoint operator whose form domain $\Dom\left(H_0^{1/2}\right)$, taken as $\hilb^+$, is a proper subspace of $\hilb$. Pick $\chi\in\hilb\setminus\hilb^+$ and take the bounded perturbation
\begin{equation}
    v(\Phi,\Psi):=\braket{\Phi,\chi}\braket{\chi,\Psi},\qquad\text{that is}\qquad V=\ket{\chi}\!\bra{\chi} .
\end{equation}
Then $|v(\Phi,\Phi)|=|\braket{\chi,\Phi}|^2\leq\|\chi\|^2\|\Phi\|^2$, so $v$ is relatively form bounded with respect to $h_0$ with relative bound $a=0$, and the hypotheses of both examples are met; in particular Hypotheses~\ref{hyp:1} and~\ref{hyp:2}(i) hold. On the other hand $V\Psi=\braket{\chi,\Psi}\chi$ belongs to $\hilb^+$ if and only if $\braket{\chi,\Psi}=0$, since $\chi\notin\hilb^+$. The set of such $\Psi$ is a closed hyperplane of $\hilb$, hence not dense, and consequently there is \emph{no} dense subspace $\mathcal{D}$ satisfying Hypothesis~\ref{hyp:2}(ii){(b)}.

What is needed, beyond relative form boundedness, is a quantitative control of the perturbation one level up in the scale. In the application of Section~\ref{sec:rwa} this is provided by two additional structural ingredients: relative \emph{operator} boundedness of the perturbation with respect to $H_0$ (Proposition~\ref{prop:rel_bound}), and a bound on the commutator $[H_0,V]$ as a map from $\hilb^+$ to $\hilb^-$ (Lemma~\ref{lem:commutator_estimate}). Together these yield $V\Dom((H_0+1)^{3/2})\subset\hilb^+$ with a quantitative estimate (Lemma~\ref{lem:potential_estimate}), which is what makes Hypothesis~\ref{hyp:2}(ii){(b)} and~{(c)} available there. We also note that the argument used in Section~\ref{sec:rwa} to obtain the uniform bound~\eqref{eq:uniform_orbit} along the orbit relies on the conservation of the energy, and is therefore specific to the case in which the generators of the original dynamics are time-independent, the time-dependence appearing only through the passage to the interaction picture, as in Example~\ref{ex:interaction}.
\end{remark}

\begin{definition}
    Let $h_{j,t}$, $j=1,2$, satisfy Hypotheses~\ref{hyp:1} and~\ref{hyp:2}. The \textit{relative action} is the time-dependent sesquilinear form $s_{21,t,s}:\hilb^+\times\hilb^+\rightarrow\mathbb{C}$ defined by
    \begin{equation}
        s_{21,t,s}(\Psi,\Phi):=\int_s^t\left[h_{2,\tau}(\Psi,\Phi)-h_{1,\tau}(\Psi,\Phi)\right]\;\mathrm{d}\tau.
    \end{equation}
\end{definition}
As with $h_{j,t}$, the relative action is uniquely associated with an operator ${S}_{21}(t,s)$ belonging to $\mathcal{B}(\hilb^+,\hilb^-)$. Indeed, Hypothesis~\ref{hyp:1} gives at once
\begin{equation}\label{eq:S21_trivial}
    \|{S}_{21}(t,s)\|_{+,-}\leq (K_1+K_2)\,|t-s| .
\end{equation}
Explicitly, ${S}_{21}(t,s)$ can be expressed as
\begin{equation}
    {S}_{21}(t,s)=\int_s^t\left[{H}_2(\tau)-{H}_1(\tau)\right]\;\mathrm{d}\tau,
\end{equation}
and satisfies
\begin{align}
    s_{21,t,s}(\Psi,\Phi)&=\left(\Psi,{S}_{21}(t,s)\Phi\right)_{+,-}\nonumber\\
    &=\int_s^t\left(\Psi,({H}_{2}(\tau)-{H}_1(\tau))\Phi\right)_{+,-}\,\mathrm{d}\tau.
\end{align}
By definition, the relative action satisfies the following properties:
\begin{lemma}\label{lem:action}
Let $\Psi,\Phi\in\hilb^+$. Then:
\begin{itemize}
    \item[(i)] $s_{21,t,s}(\Psi,\Phi)=-s_{12,t,s}(\Psi,\Phi)$;
    \item[(ii)] $s_{21,t,s}(\Phi,\Psi)=\overline{s_{21,t,s}(\Psi,\Phi)}$;
    \item[(iii)] $\dfrac{\mathrm{d}}{\mathrm{d}t}s_{21,t,s}(\Psi,\Phi)=h_{2,t}(\Psi,\Phi)-h_{1,t}(\Psi,\Phi)=\left(\Psi,\left[{H}_2(t)-{H}_1(t)\right]\Phi\right)_{+,-}$;
\item[(iv)] $|s_{21,t,s}(\Psi,\Phi)|\leq\|{S}_{21}(t,s)\|_{+,-}\|\Psi\|_+\|\Phi\|_+$.
\end{itemize}

\end{lemma}
\begin{proof}
    The first property is immediate from the definition, and the second follows from the Hermiticity assumed in Hypothesis~\ref{hyp:1}. The third is a direct consequence of the fundamental theorem of calculus. The last bound follows from
    \begin{align}
        |s_{21,t,s}(\Psi,\Phi)|&=\left|\left(\Psi,{S}_{21}(t,s)\Phi\right)_{+,-}\right|\nonumber\\
        &\leq\|\Psi\|_+\left\|{S}_{21}(t,s)\Phi\right\|_-\nonumber\\
        &\leq\|\Psi\|_+\|{S}_{21}(t,s)\|_{+,-}\|\Phi\|_+,
    \end{align}
    where the first inequality above is the canonical extension of the Cauchy--Bunyakovsky--Schwarz inequality to the duality pairing, cf.\ \cite{Berezanskii1968}.
\end{proof}

The following elementary consequence of Hypothesis~\ref{hyp:2} will be used to justify the differentiations performed in the proof of Proposition~\ref{prop:equality}. It upgrades the differentiability of $\tau\mapsto U_j(\tau,s)\Psi$ from the topology of $\hilb^-$, in which it is postulated, to that of $\hilb^+$.

\begin{lemma}\label{lem:plus_differentiable}
    Let Hypotheses~\ref{hyp:1} and~\ref{hyp:2} hold and let $\Psi\in\mathcal{D}$. Then, for $j=1,2$, the map $\tau\mapsto U_j(\tau,s)\Psi$ is differentiable from $I$ to $\hilb^+$, and
    \begin{equation}\label{eq:plus_schro}
        \ii\frac{\mathrm{d}}{\mathrm{d}\tau}U_j(\tau,s)\Psi={H}_j(\tau)U_j(\tau,s)\Psi
    \end{equation}
    holds in $\hilb^+$.
\end{lemma}
\begin{proof}
    Set $g(\tau):=U_j(\tau,s)\Psi$ and $w(\tau):=-\ii{H}_j(\tau)U_j(\tau,s)\Psi$. By Hypothesis~\ref{hyp:2}(ii){(a)}--{(b)}, $g$ and $w$ take values in $\hilb^+$; by {(c)}, $w$ is continuous from $I$ to $\hilb^+$, so that the Bochner integral
    \begin{equation}
        \tilde g(\tau):=\Psi+\int_s^\tau w(\sigma)\,\mathrm{d}\sigma
    \end{equation}
    is well defined in $\hilb^+$ and defines a map which is differentiable from $I$ to $\hilb^+$ with $\tilde g'=w$. On the other hand, by Hypothesis~\ref{hyp:2}(i) the map $g$ is differentiable from $I$ to $\hilb^-$ with $g'=w$, whence $g(\tau)=\Psi+\int_s^\tau w(\sigma)\,\mathrm{d}\sigma$ as an identity in $\hilb^-$. Since the two integrals coincide, the continuous embedding $\hilb^+\hookrightarrow\hilb^-$ being injective, we conclude $g=\tilde g$, and~\eqref{eq:plus_schro} follows.
\end{proof}

\begin{proposition}\label{prop:equality}
        Let $h_{j,t}$, $j=1,2$, be time-dependent sesquilinear forms satisfying Hypotheses~\ref{hyp:1} and~\ref{hyp:2}. For $\Phi,\Psi\in\mathcal{D}$, cf.\ Hypothesis~\ref{hyp:2}(ii), the following identity holds:
        \begin{align}
            \ii\braket{U_1(t,s)\Phi,\left[U_2(t,s)-U_1(t,s)\right]\Psi}=&s_{21,t,s}\left(U_1(t,s)\Phi,U_2(t,s)\Psi\right)\nonumber\\
            &+\ii\int_s^t\bigg(s_{21,\tau,s}\left(U_1(\tau,s)\Phi,{H}_2(\tau)U_2(\tau,s)\Psi\right)\nonumber\\&-s_{21,\tau,s}\left({H}_1(\tau)U_1(\tau,s)\Phi,U_2(\tau,s)\Psi\right)\bigg)\;\mathrm{d}\tau.
        \end{align}
\end{proposition}

\begin{proof}
    First notice that all terms in the integral are well-defined: by Hypothesis~\ref{hyp:2}(ii){(a)}--{(b)} one has $U_j(\tau,s)\mathcal{D}\subset\mathcal{D}$ and ${H}_j(\tau)\mathcal{D}\subset\hilb^+$, so that ${H}_1(\tau)U_1(\tau,s)\Phi\in\hilb^+$ and ${H}_2(\tau)U_2(\tau,s)\Psi\in\hilb^+$.

    The proof follows the strategy of \cite{burgarth2022one}, adapted to the form setting. We begin by establishing the auxiliary identity~\eqref{eq:lemma5}; here the argument departs from the operator-level one, for the reason explained after~\eqref{eq:lemma5}. By Lemma~\ref{lem:plus_differentiable} the maps $\tau\mapsto U_1(\tau,s)\Phi$ and $\tau\mapsto U_2(\tau,s)\Psi$ are differentiable from $I$ to $\hilb^+$, hence in particular from $I$ to $\hilb$, with derivatives $-\ii{H}_1(\tau)U_1(\tau,s)\Phi$ and $-\ii{H}_2(\tau)U_2(\tau,s)\Psi$, both taking values in $\hilb^+$. The ordinary product rule for the scalar product of two differentiable $\hilb$-valued maps therefore gives
    \begin{equation}\label{eq:scalar_derivative}
        \frac{\mathrm{d}}{\mathrm{d}\tau}\braket{U_1(\tau,s)\Phi,U_2(\tau,s)\Psi}
        =\ii\braket{{H}_1(\tau)U_1(\tau,s)\Phi,U_2(\tau,s)\Psi}
        -\ii\braket{U_1(\tau,s)\Phi,{H}_2(\tau)U_2(\tau,s)\Psi}.
    \end{equation}
    Note that all the vectors appearing in~\eqref{eq:scalar_derivative}, and in the identities below, belong to $\hilb^+\subset\hilb$, so that every pairing occurring in this proof is an ordinary scalar product of $\hilb$ and no duality pairing is needed. Using the Hermiticity assumed in Hypothesis~\ref{hyp:1},
    \begin{equation}
        \braket{{H}_1(\tau)U_1(\tau,s)\Phi,U_2(\tau,s)\Psi}
        =\overline{h_{1,\tau}\left(U_2(\tau,s)\Psi,U_1(\tau,s)\Phi\right)}
        =\braket{U_1(\tau,s)\Phi,{H}_1(\tau)U_2(\tau,s)\Psi},
    \end{equation}
    where ${H}_1(\tau)U_2(\tau,s)\Psi\in\hilb^+$ by Hypothesis~\ref{hyp:2}(ii){(a)}--{(b)}. Hence~\eqref{eq:scalar_derivative} becomes
    \begin{equation}\label{eq:scalar_derivative2}
        \frac{\mathrm{d}}{\mathrm{d}\tau}\braket{U_1(\tau,s)\Phi,U_2(\tau,s)\Psi}
        =-\ii\braket{U_1(\tau,s)\Phi,\left[{H}_2(\tau)-{H}_1(\tau)\right]U_2(\tau,s)\Psi}.
    \end{equation}
    The right-hand side is continuous in $\tau$, by Hypothesis~\ref{hyp:2}(ii){(c)} and by the continuity of $\tau\mapsto U_1(\tau,s)\Phi$ in $\hilb^+$. Integrating~\eqref{eq:scalar_derivative2} from $s$ to $t$, and using $U_j(s,s)=I$ together with the unitarity of $U_1(t,s)$ in the form $\braket{U_1(t,s)\Phi,U_1(t,s)\Psi}=\braket{\Phi,\Psi}$, we obtain
    \begin{align}\label{eq:lemma5}
        \ii\braket{U_1(t,s)\Phi,\left[U_2(t,s)-U_1(t,s)\right]\Psi}
        =\int_s^t\braket{U_1(\tau,s)\Phi,\left[{H}_2(\tau)-{H}_1(\tau)\right]U_2(\tau,s)\Psi}\;\mathrm{d}\tau.
    \end{align}
    We stress that this argument never differentiates a propagator with respect to its second time argument, never applies a propagator to a vector of $\hilb^-$, and in fact never leaves $\hilb^+$; only the $\hilb^+$-valued differentiability of Lemma~\ref{lem:plus_differentiable} and the Hermiticity of the forms are used. For comparison, the operator-level derivation, which proceeds through the identity
    $U_2(t,s)-U_1(t,s)=U_1(t,s)\left[U_1(s,t)U_2(t,s)-I\right]$ and differentiates $\tau\mapsto U_1(s,\tau)U_2(\tau,s)$, is not available here: it would require both of those operations, neither of which follows from Hypothesis~\ref{hyp:2}.

We next carry out the form-level counterpart of the second integration by parts of \cite{burgarth2022one}, namely the differentiation of the map

\begin{equation}
\tau\longmapsto s_{21,\tau,s}\left(U_1(\tau,s)\Phi,U_2(\tau,s)\Psi\right) = \left(U_1(\tau,s)\Phi,\;{S}_{21}(\tau,s)U_2(\tau,s)\Psi\right)_{+,-}.
\end{equation}

Three sources of $\tau$-dependence appear. The dependence through the form $s_{21,\tau,s}$ is handled by Lemma~\ref{lem:action}\textit{(iii)}, which applies by the continuity assumed in Hypothesis~\ref{hyp:1}. The two arguments are differentiable in the topology of $\hilb^+$ by Lemma~\ref{lem:plus_differentiable}; this is what is needed, rather than mere differentiability in $\hilb^-$, since the derivative of the first argument has to be paired with ${S}_{21}(\tau,s)U_2(\tau,s)\Psi\in\hilb^+$, and symmetrically for the second. That the latter vector does lie in $\hilb^+$ follows from Hypothesis~\ref{hyp:2}(ii){(c)}, which moreover bounds it uniformly for $\tau$ ranging in the interval $[s,t]$ over which the integration is performed: for $s\leq\tau\leq t$,
\begin{equation}
 \left\|{S}_{21}(\tau,s)U_j(\tau,s)\Psi\right\|_+ \leq \int_s^\tau\left\|\left[{H}_2(\tau')-{H}_1(\tau')\right]U_j(\tau,s)\Psi\right\|_+\d\tau'\leq 2M_\Psi(\tau-s)\leq 2M_\Psi(t-s).
\end{equation}
Applying the product rule to the three factors just discussed, we obtain 
    \begin{align}\label{eq:lemma6}
        \frac{\mathrm{d}}{\mathrm{d}\tau}s_{21,\tau,s}\left(U_1(\tau,s)\Phi,U_2(\tau,s)\Psi\right)=&\braket{U_1(\tau,s)\Phi,\left[{H}_2(\tau)-{H}_1(\tau)\right]U_2(\tau,s)\Psi}\nonumber\\
        &+\ii\Big[s_{21,\tau,s}\left({H}_1(\tau)U_1(\tau,s)\Phi,U_2(\tau,s)\Psi\right)\nonumber\\&-s_{21,\tau,s}\left(U_1(\tau,s)\Phi,{H}_2(\tau)U_2(\tau,s)\Psi\right)\Big].
    \end{align}
    Using Eq.~\eqref{eq:lemma6} in Eq.~\eqref{eq:lemma5} proves the statement.
\end{proof}

\begin{remark}
    The operator-level analogue of the equation in \cite{burgarth2022one},
    \begin{align}
    \left(U_2(t,s)-U_1(t,s)\right)\Psi=&-\ii S_{21}(t,s)U_2(t,s)\Psi\nonumber\\&-U_1(t,s)\int_s^tU_1^\dag(\tau,s)\left[{H}_1(\tau){S}_{21}(\tau,s)-{S}_{21}(\tau,s){H}_2(\tau)\right]U_2(\tau,s)\Psi\;\mathrm{d}\tau,
\end{align}
with $\Psi\in\mathcal{D}$, does \textit{not} hold in general. The term ${S}_{21}(\tau,s){H}_2(\tau)U_2(\tau,s)\Psi$ is well defined, since ${H}_2(\tau)U_2(\tau,s)\Psi\in\hilb^+$ by Hypothesis~\ref{hyp:2}(ii) and ${S}_{21}(\tau,s)$ maps $\hilb^+$ into $\hilb^-$. The expression ${H}_1(\tau){S}_{21}(\tau,s)U_2(\tau,s)\Psi$, on the other hand, need not be: the vector ${S}_{21}(\tau,s)U_2(\tau,s)\Psi$ lies in $\hilb^+$, but ${H}_1(\tau)$ is only defined on $\hilb^+$ as a map \emph{into} $\hilb^-$, so that a further application of ${S}_{21}$, or of ${H}_1$, is not available. Making sense of the expression would require the stronger property that ${H}_1(\tau')$ map ${S}_{21}(\tau,s)U_2(\tau,s)\Psi$ back into $\hilb^+$ for every intermediate time $s\leq\tau'\leq\tau$, and this does not follow from parts~{(a)}--{(c)} of Hypothesis~\ref{hyp:2}(ii), which control ${H}_j$ only on $\mathcal{D}$ and along the orbits of the propagators. It is precisely in order to bypass this difficulty that the identity of Proposition~\ref{prop:equality} is formulated at the level of the forms, where only the duality pairing between $\hilb^+$ and $\hilb^-$ is needed.
\end{remark}

\begin{corollary}\label{cor:equality}
    Let $\Phi,\Psi\in\mathcal{D}$. Then
    \begin{align}
            \ii\braket{\Phi,\left[U_2(t,s)-U_1(t,s)\right]\Psi}=&s_{21,t,s}\left(\Phi,U_2(t,s)\Psi\right)\nonumber\\
            &+\ii\int_s^t\Big[s_{21,\tau,s}\left(U_1(\tau,t)\Phi,{H}_2(\tau)U_2(\tau,s)\Psi\right)\nonumber\\&-s_{21,\tau,s}\left({H}_1(\tau)U_1(\tau,t)\Phi,U_2(\tau,s)\Psi\right)\Big]\;\mathrm{d}\tau.
        \end{align}
\end{corollary}

\begin{proof}
    Apply Proposition~\ref{prop:equality} with $\Phi$ replaced by $U_1(s,t)\Phi$, which belongs to $\mathcal{D}$ by Hypothesis~\ref{hyp:2}(ii){(a)}.
\end{proof}

\begin{theorem}\label{thm:bound}
    Let $h_{j,t}$, $j=1,2$, be time-dependent sesquilinear forms satisfying Hypotheses~\ref{hyp:1} and~\ref{hyp:2}. For $s\in I$ and $\Phi\in\mathcal{D}$, cf.\ Hypothesis~\ref{hyp:2}(ii),
     \begin{align}\label{eq:thebound}
          \frac{1}{2} \left\|\left[U_2(t,s)-U_1(t,s)\right]\Phi\right\|^2\leq& \;\|{S}_{21}(t,s)\|_{+,-}\|U_2(t,s)\Phi\|_+\|U_1(t,s)\Phi\|_+\nonumber\\
           &+\int_s^t\|{S}_{21}(\tau,s)\|_{+,-}\bigg(\|U_2(\tau,s)\Phi\|_+\|{H}_1(\tau)U_1(\tau,s)\Phi\|_+\nonumber\\
          &\qquad\qquad\qquad\qquad+\|U_1(\tau,s)\Phi\|_+\|{H}_2(\tau)U_2(\tau,s)\Phi\|_+\bigg)\;\mathrm{d}\tau.
        \end{align}
\end{theorem}
\begin{proof}
    Applying Proposition~\ref{prop:equality} with $\Phi=\Psi\in\mathcal{D}$ gives
    \begin{align}
            \ii\braket{U_1(t,s)\Phi,\left[U_2(t,s)-U_1(t,s)\right]\Phi}=&s_{21,t,s}\left(U_1(t,s)\Phi,U_2(t,s)\Phi\right)\nonumber\\
            &+\ii\int_s^t\Big[s_{21,\tau,s}\left(U_1(\tau,s)\Phi,{H}_2(\tau)U_2(\tau,s)\Phi\right)\nonumber\\&-s_{21,\tau,s}\left({H}_1(\tau)U_1(\tau,s)\Phi,U_2(\tau,s)\Phi\right)\Big]\;\mathrm{d}\tau.
        \end{align}
        Applying Proposition~\ref{prop:equality} with the two propagators interchanged gives
        \begin{align}
            \ii\braket{U_2(t,s)\Phi,\left[U_1(t,s)-U_2(t,s)\right]\Phi}=&s_{12,t,s}\left(U_2(t,s)\Phi,U_1(t,s)\Phi\right)\nonumber\\
            &+\ii\int_s^t\Big[s_{12,\tau,s}\left(U_2(\tau,s)\Phi,{H}_1(\tau)U_1(\tau,s)\Phi\right)\nonumber\\&-s_{12,\tau,s}\left({H}_2(\tau)U_2(\tau,s)\Phi,U_1(\tau,s)\Phi\right)\Big]\;\mathrm{d}\tau,
        \end{align}
        and using $s_{12,\tau,s}=-s_{21,\tau,s}$,
        \begin{align}
            \ii\braket{U_2(t,s)\Phi,\left[U_2(t,s)-U_1(t,s)\right]\Phi}=&s_{21,t,s}\left(U_2(t,s)\Phi,U_1(t,s)\Phi\right)\nonumber\\
            &+\ii\int_s^t\Big[s_{21,\tau,s}\left(U_2(\tau,s)\Phi,{H}_1(\tau)U_1(\tau,s)\Phi\right)\nonumber\\&-s_{21,\tau,s}\left({H}_2(\tau)U_2(\tau,s)\Phi,U_1(\tau,s)\Phi\right)\Big]\;\mathrm{d}\tau.
        \end{align}
        Subtracting the first identity from the second and using $s_{21,t,s}(\Phi,\Psi)=\overline{s_{21,t,s}(\Psi,\Phi)}$, which gives
        \begin{align}
            s_{21,t,s}(\Psi,\Phi)-s_{21,t,s}(\Phi,\Psi)&=2\ii\operatorname{Im} s_{21,t,s}(\Psi,\Phi),\\
            s_{21,t,s}(\Psi,\Phi)+s_{21,t,s}(\Phi,\Psi)&=2\operatorname{Re} s_{21,t,s}(\Psi,\Phi),
        \end{align}
        yields the identity
        \begin{align}
          \frac{1}{2} \left\|\left[U_2(t,s)-U_1(t,s)\right]\Phi\right\|^2=& \operatorname{Im} s_{21,t,s}\left(U_2(t,s)\Phi,U_1(t,s)\Phi\right)\nonumber\\
           &+\int_s^t\operatorname{Re}s_{21,\tau,s}\left(U_2(\tau,s)\Phi,{H}_1(\tau)U_1(\tau,s)\Phi\right)\;\mathrm{d}\tau\nonumber\\
           &-\int_s^t\operatorname{Re}s_{21,\tau,s}\left(U_1(\tau,s)\Phi,{H}_2(\tau)U_2(\tau,s)\Phi\right)\;\mathrm{d}\tau,
        \end{align}
from which~\eqref{eq:thebound} follows upon estimating each term by Lemma~\ref{lem:action}\textit{(iv)}.
\end{proof}

\section{Application: the rotating-wave approximation}\label{sec:rwa}

\subsection{Spin--boson models}\label{sec:prelim}
Let $\sps=L^2_\mu(X)$ be the single-particle space of a boson field, i.e., the space of square-integrable functions on a measure space $(X,\Sigma,\mu)$ representing the field's momentum space. The associated Bose--Fock space $\fock(\sps)\equiv\fock$~\cite{bratteli1997,reed1975ii} consists of completely symmetric sequences $\psi=\{\psi^{(n)}\}_{n\in\mathbb{N}}$, with $\psi^{(n)}$ an $n$-boson wavefunction, satisfying
\begin{equation}\label{eq:spin-boson1}
	\sum_{n\in\mathbb{N}}\int\mathrm{d}^n\mu\;\left|\psi^{(n)}(k_1,\dots,k_n)\right|^2=:\|\psi\|^2_\fock<\infty.
\end{equation}

Given a measurable dispersion relation $\omega:X\to\mathbb{R}$ with $\omega\geq m>0$, the free boson energy is the \textit{second quantization} $\dOmega$ of $\omega$, a nonnegative self-adjoint operator on $\fock$, defined as
\begin{equation}\label{eq:fieldenergy}
	\dOmega=\int\omega(k)a^\dag_ka_k\;\mathrm{d}\mu,
\end{equation}
where $a_k,a_k^\dag$ are operator-valued distributions satisfying the bosonic canonical commutation relations, and $m$ is the \textit{mass} (lowest frequency) of the field. 
Consider a two-level system (spin) with an energy gap $\omega_0\geq0$ between its excited state and its ground state, the latter being set to zero, coupled to a boson field on the Bose--Fock space $\fock$. Let $\hfrak=\mathbb{C}^2\otimes\fock$ be the Hilbert space for a composite system of a two-level system and the bosonic field.

For a form factor $f:X\rightarrow\mathbb{C}$ satisfying
\begin{equation}\label{eq:minusone}
	\|f\|_{-1}^2:=\int\omega(k)^{-1}|f(k)|^2\,\mathrm{d}\mu=\|\omega^{-1/2}f\|^2<\infty,
\end{equation}
one defines
\begin{eqnarray}
	H&=&H_{\rm free}+\,(\sigma_+ + \sigma_-)\otimes\left(\a{f}+\adag{f}\right),\label{eq:hamiltonian}\\
	H_{\rm RWA}&=&H_{\rm free}+\left(\sigma_+\otimes\a{f}+\sigma_-\otimes\adag{f}\right),\label{eq:hamiltonian_rwa}
\end{eqnarray}
where $H_{\rm free}$ is the free Hamiltonian of the two-level system coupled to the bosonic field and is given by
\begin{equation}
	H_{\rm free}=S\otimes\mathrm{I}+\mathrm{I}\otimes\dOmega.
\end{equation} 
Here, $S\in\mathcal{B}(\mathbb{C}^2)$, $S=S^\dag\geq0$, is the spin Hamiltonian, which without loss of generality will be taken to be diagonal with the lowest eigenvalue set to zero, so that $S=\omega_0 K$ with $K=\begin{bmatrix} 1 & 0 \\ 0 & 0 \end{bmatrix}$. The operators $\a{f},\adag{f}$ are the annihilation and creation operators associated with the function $f$, formally given by
\begin{equation}
    \a{f}=\int \overline{f(k)}a_k\;\mathrm{d}\mu,\qquad \adag{f}=\int f(k)a^\dag_k\;\mathrm{d}\mu,
\end{equation}
see Proposition~\ref{prop:af} below. 
Finally, the matrices $\sigma_\pm$ are given by 
\begin{equation}\label{eq:sigmas}
\sigma_+= \begin{bmatrix} 0 & 1 \\ 0 & 0 \end{bmatrix},\qquad \sigma_-=\sigma_+^\dagger= \begin{bmatrix} 0 & 0 \\ 1 & 0 \end{bmatrix}.
\end{equation}

Since we are interested in the regime where both the field energy and the spin gap diverge, we introduce a dimensionless parameter $\alpha$ to study the limit of large frequency, $\alpha\to\infty$:
\begin{equation}
    H_{\rm free}=\begin{bmatrix}
        \alpha\omega_0+\Delta&0\\0&0
    \end{bmatrix}\otimes\mathrm{I}+\mathrm{I}\otimes\alpha\,\dOmega,
\end{equation}
with detuning $\Delta\geq0$, a parameter independent of $\alpha$ (the frequency rescaling). Note that $\alpha\dOmega=\mathrm{d}\Gamma(\alpha\omega)$. We split the free Hamiltonian as
\begin{equation}\label{eq:spin-boson5}
    H_{\rm free}=\alpha H_0+\Delta K\otimes\mathrm{I},\qquad H_0=\omega_0 K\otimes\mathrm{I}+\mathrm{I}\otimes\dOmega,
\end{equation}
with $K$ as above.

Under~\eqref{eq:minusone} alone the expressions~\eqref{eq:hamiltonian}--\eqref{eq:hamiltonian_rwa} are to be interpreted as sesquilinear forms, the smeared creation and annihilation operators being defined as maps from $\fock^+$ to $\fock^-$; see \cite{lonigro2022generalized} for details, and \cite{spohn1989,arai1997} for the mathematical study of models of this type.

The operators $\omega,\dOmega$ and $H_{0}$ are all self-adjoint and nonnegative on $\sps$, $\fock$, $\hfrak$ respectively, and therefore each defines a \textit{scale of Hilbert spaces}:
\begin{equation}
\sps^+\subset\sps\subset\sps^-,\qquad\fock^+\subset\fock\subset\fock^-,\qquad\hfrak^+\subset\hfrak\subset\hfrak^-,
\end{equation}
with $\sps^+$, $\fock^+$, $\hfrak^+$ defined as the form domains of $\omega,\dOmega, H_{0}$ endowed with the norms
\begin{eqnarray}
	\|f\|_{\sps^+}&:=&\|\omega^{1/2}f\|_\sps,\\ \|\psi\|_{\fock^+}&:=&\|(\dOmega + 1)^{1/2}\psi\|_\fock,\\
	\|\Psi\|_{+,\alpha}&:=&\|(\alpha H_{0} + \alpha)^{1/2}\Psi\|_\hfrak,\label{eq:+_norm}
\end{eqnarray}
and $\sps^-,\fock^-,\hfrak^-$ being their corresponding dual spaces; we write $\|\cdot\|_{\sps^-}$, $\|\cdot\|_{\fock^-}$, $\|\cdot\|_{-,\alpha}$ for the respective dual norms. We will use the notation $\|\cdot\|_{+}=\|\cdot\|_{+,1}$ for the $\alpha$-independent norm.

Condition~\eqref{eq:minusone} on the form factor, under which the model is well defined~\cite{lonigro2022generalized}, is all that is required for the definition of the model and for the elementary properties of the smeared creation and annihilation operators recalled in this subsection; the results from Proposition~\ref{prop:rel_bound} onwards require the generally stronger Hypothesis~\ref{hyp:f}, stated below. Under~\eqref{eq:minusone} the smeared operators extend to bounded maps between the Sobolev-type spaces $\fock^\pm$:

\begin{proposition}[\cite{lonigro2022generalized}, Props. 3.4--3.5]\label{prop:af}
	Let $\omega\geq m$ be a real-valued measurable function, and $f\in\sps^-$ (i.e.\ $\|f\|_{-1}<\infty$). Then:
	\begin{itemize}
		\item[(i)] the annihilation operator smeared against $f$ is a continuous operator $\a{f}\in\mathcal{B}(\fock^{+},\fock)$;
		\item[(ii)] its adjoint $\adag{f}:=\a{f}^*\in\mathcal{B}(\fock,\fock^{-})$ with respect to the duality pairing between $\fock^+$ and $\fock^{-}$ is a continuous operator whose action on $\psi\in\fock^{+}$ agrees with the creation operator.
	\end{itemize}
The operator norms satisfy
\begin{equation}
	\left\|\a{f}\right\|_{\mathcal{B}(\fock^{+},\fock)}\leq\|f\|_{-1},\qquad \left\|\adag{f}\right\|_{\mathcal{B}(\fock,\fock^{-})}\leq\|f\|_{-1}.
\end{equation}
\end{proposition}
Consequently, using the inclusions $\fock^+\subset\fock\subset\fock^-$, one also has $\a{f},\adag{f}\in\mathcal{B}\left(\fock^+,\fock^-\right)$ with 
\begin{equation}\label{eq:normbounds}
	\left\|\a{f}\right\|_{\mathcal{B}(\fock^{+},\fock^-)}\leq\|f\|_{-1},\qquad \left\|\adag{f}\right\|_{\mathcal{B}(\fock^+,\fock^{-})}\leq\|f\|_{-1}.
\end{equation}
Interpreting $\a{f},\adag{f}$ as operators from $\fock^+$ to $\fock^-$, the expressions
\begin{eqnarray}\label{eq:h}
	H&=&\alpha H_0+\Delta K\otimes\mathrm{I}+(\sigma_++\sigma_-)\otimes\left(\a{f}+\adag{f}\right),\\\label{eq:hrwa}
	H_{\rm RWA}&=&\alpha H_0+\Delta K\otimes\mathrm{I}+\left(\sigma_+\otimes\a{f}+\sigma_-\otimes\adag{f}\right),
\end{eqnarray}
with $\sigma_\pm$ as in Eq.~\eqref{eq:sigmas}, define two operators from $\mathbb{C}^2\otimes\fock^+$ to $\mathbb{C}^2\otimes\fock^-$.

In what follows we will denote by $H_1$ and $V_1$ the Hamiltonian and interaction term, respectively, of the exact dynamics, and by $H_2$ and $V_2$ the corresponding operators for the rotating-wave approximation. The detuning factor $\Delta$, being $\alpha$-independent and determining a bounded perturbation of the Hamiltonians does not play any role in the forthcoming discussion. For simplicity we will consider $\Delta=0$ from now on. All the results are also true for $\Delta\neq0$ but the different constants would be modified by some factors that depend on this parameter. 

\begin{lemma}\label{lem:scale_norms}
Let $\|\cdot\|_{\pm, \alpha}$ be as in Eq.~\eqref{eq:+_norm}, recall that $\|\cdot\|_{\pm} = \|\cdot\|_{\pm, 1}$, and let $\|\cdot\|_{\mathbb{C}^2\otimes\fock^{\pm}}$ denote the canonical norms in the tensor product spaces. Then
  \begin{enumerate}
  \item $\|\cdot\|_{+, \alpha} = \sqrt{\alpha}\,\| \cdot \|_+$,
  \item $\|\cdot\|_{-, \alpha} = \frac{1}{\sqrt{\alpha}}\| \cdot \|_-$,
  \item $\|\cdot\|_{\mathbb{C}^2\otimes\fock^{+}} \leq \|\cdot\|_{+} \leq \sqrt{1+\omega_0}\,\|\cdot\|_{\mathbb{C}^2\otimes\fock^{+}}$,
  \item $\frac{1}{\sqrt{1+\omega_0}}\|\cdot\|_{\mathbb{C}^2\otimes\fock^{-}} \leq \|\cdot\|_{-} \leq \|\cdot\|_{\mathbb{C}^2\otimes\fock^{-}}$.
  \end{enumerate}
\end{lemma}

\begin{proof}
The proofs of \textit{(1)} and \textit{(2)} are immediate from the definition. To prove \textit{(3)} it is enough to show it for an element of the form $\Psi = u\otimes\psi$. On the one hand we have
\begin{align}
\| u\otimes\psi\|^2_{\mathbb{C}^2\otimes\fock^{+}} = \|u\|^2_{\mathbb{C}^2}\|(\dOmega+1)^{1/2} \psi\|_\fock^2 
  &=  \|u\|^2_{\mathbb{C}^2}\langle \psi,(\dOmega+1)\psi\rangle_\fock \\
  & \leq \langle u, \omega_0 Ku \rangle \|\psi\|^2_\fock + \|u\|^2_{\mathbb{C}^2}\langle \psi,(\dOmega+1)\psi\rangle_\fock \\
  &= \|u\otimes\psi\|_+^2 .
\end{align}
On the other hand
\begin{align}
\|u\otimes\psi\|_+^2 
  &= \langle u, \omega_0 Ku \rangle \|\psi\|^2_\fock + \|u\|^2_{\mathbb{C}^2}\langle \psi,(\dOmega+1)\psi\rangle_\fock \\
  &\leq \omega_0\|u\|^2_{\mathbb{C}^2}\|\psi\|^2_\fock +\|u\|^2_{\mathbb{C}^2}\langle \psi,(\dOmega+1)\psi\rangle_\fock\\
  &\leq \left(1 + \omega_0\right)\|u\|^2_{\mathbb{C}^2}\langle \psi,(\dOmega+1)\psi\rangle_\fock\\
  &= \left(1 + \omega_0\right)\| u\otimes\psi\|^2_{\mathbb{C}^2\otimes\fock^{+}}.
\end{align}
The proof of \textit{(4)} follows by standard arguments in the scale of Hilbert spaces relating the norms of the dual spaces.
\end{proof}

The results of this subsection, starting with Proposition~\ref{prop:rel_bound}, require the following strengthening of~\eqref{eq:minusone}.

\begin{hypothesis}\label{hyp:f}
The form factor satisfies $f\in\sps^+$, that is,
\begin{equation}\label{eq:fplus}
    \|f\|_{+1}^2:=\int\omega(k)\,|f(k)|^2\,\mathrm{d}\mu<\infty .
\end{equation}
\end{hypothesis}

\begin{remark}\label{rem:fplus}
Since $\omega\geq m>0$, Hypothesis~\ref{hyp:f} implies the two weaker conditions that are also used below. Indeed, writing $\|f\|^2_{L^2}:=\int|f(k)|^2\,\mathrm{d}\mu$ and using $\omega^{-1}\leq m^{-2}\omega$ and $1\leq m^{-1}\omega$ pointwise,
\begin{equation}\label{eq:fchain}
    \|f\|_{-1}\leq \frac{1}{m}\|f\|_{+1},\qquad \|f\|_{L^2}\leq \frac{1}{\sqrt m}\|f\|_{+1}.
\end{equation}
In particular $f\in L^2_\mu(X)$ and $\|f\|_{-1}<\infty$. Note also that $\|\omega f\|_{-1}=\|f\|_{+1}$, which is precisely the quantity appearing in the commutator estimate of Lemma~\ref{lem:commutator_estimate}.
\end{remark}

We shall start by introducing some results that guarantee that the problem is well-posed and that will be of later use.

\begin{proposition}\label{prop:rel_bound}
The operators $V_j$, $j=1,2$, are relatively form bounded and relatively operator bounded with respect to $\alpha H_0$ and in either case there exists $\alpha_0$ such that the relative bound can be taken smaller than one for all $\alpha>\alpha_0$.
\end{proposition}

\begin{proof}
  Notice that by Proposition~\ref{prop:af} $\left\| a^\sharp({f})\right\|_{\mathcal{B}(\fock^{+},\fock^-)} \leq \|f\|_{-1}$. We have
  \begin{align}
    |\langle \Phi, V_j \Phi \rangle | &\leq |\langle \Phi, \sigma^\sharp\otimes a({f}) \Phi \rangle | + |\langle \Phi, \sigma^\sharp\otimes a^\dagger({f}) \Phi \rangle |\\
      & \leq 2\|f\|_{-1}\|\Phi\|^2_{\mathbb{C}^2\otimes \fock^+}\\
      & \leq 2\|f\|_{-1}\|\Phi\|^2_{+}\\
      & = 2\|f\|_{-1}\langle \Phi, (H_0+1) \Phi \rangle = \frac{2\|f\|_{-1}}{\alpha} \langle \Phi, (\alpha H_0+\alpha) \Phi \rangle    
  \end{align}
  where we have used the equivalence of the tensor product norm and the norm $\|\cdot\|_+$ established in Lemma~\ref{lem:scale_norms}\textit{(3)}.
  Here $\sigma^\sharp$ denotes either $\sigma_++\sigma_-$ (for $j=1$) or $\sigma_\pm$ (for $j=2$); in all cases $\|\sigma^\sharp\|_{\mathcal{B}(\mathbb{C}^2)}=1$, which is why the interaction contributes the factor $2$ rather than the number of its summands.

  For the relative operator boundedness we use the standard Nelson-type estimates
  \begin{equation}\label{eq:nelson}
      \|\a{f}\psi\|_\fock\leq\|f\|_{-1}\|\dOmega^{1/2}\psi\|_\fock,\qquad \|\adag{f}\psi\|_\fock\leq\|f\|_{-1}\|\dOmega^{1/2}\psi\|_\fock+\|f\|_{L^2}\|\psi\|_\fock,
  \end{equation}
  which follow from Proposition~\ref{prop:af} together with the canonical commutation relations, and where $\|f\|_{L^2}$ is finite by Remark~\ref{rem:fplus}. Interpolating in~\eqref{eq:nelson},
  \begin{equation}
      \|\dOmega^{1/2}\psi\|^2_\fock=\langle\psi,\dOmega\psi\rangle_\fock\leq\varepsilon\|\dOmega\psi\|^2_\fock+\frac{1}{4\varepsilon}\|\psi\|_\fock^2 ,
  \end{equation}
  so that $\|\dOmega^{1/2}\psi\|_\fock\leq\sqrt\varepsilon\,\|\dOmega\psi\|_\fock+\frac{1}{2\sqrt\varepsilon}\|\psi\|_\fock$ for every $\varepsilon>0$. Since $H_0=\omega_0K\otimes\mathrm{I}+\mathrm{I}\otimes\dOmega$ with $\|\omega_0K\|_{\mathcal{B}(\mathbb{C}^2)}=\omega_0$, one has $\|\mathrm{I}\otimes\dOmega\,\Psi\|\leq\|H_0\Psi\|+\omega_0\|\Psi\|$, and combining the three displays we obtain, for every $\varepsilon>0$,
  
\begin{equation}
\|V_j\Psi\| \leq 2\sqrt\varepsilon \|f\|_{-1} \|H_0\Psi\| + \left(\frac{\|f\|_{-1}}{\sqrt\varepsilon} + \|f\|_{L^2} + 2\sqrt{\varepsilon}\,\omega_0\|f\|_{-1}\right)\|\Psi\| .
\end{equation}

Choosing, for instance, $\sqrt\varepsilon=\left(4\|f\|_{-1}+1\right)^{-1}$, so that $2\sqrt\varepsilon\|f\|_{-1}\leq\tfrac12$, we obtain constants $a\in(0,1)$ and $b>0$ depending only on $\omega_0$, $\|f\|_{-1}$ and $\|f\|_{L^2}$, and in particular independent of $\alpha$, such that  
\begin{equation}
\|V_j\Psi\| \leq \frac{a}{\alpha}\|\alpha H_0\Psi\| + b\|\Psi\|, 
\end{equation}
thus concluding the proof.
\end{proof}

\begin{proposition}\label{prop:sa_and_bound}
  Let $\alpha>2\|f\|_{-1}$. Then the operators $H_j$, $j=1,2$, are self-adjoint and bounded from below; more precisely
  \begin{equation}\label{eq:semibounded}
      H_j\geq -2\|f\|_{-1},\qquad\text{so that}\qquad H_j+\alpha\geq \alpha-2\|f\|_{-1}>0 .
  \end{equation}
\end{proposition}

\begin{proof}
Self-adjointness is a direct application of the KLMN theorem, cf.\ \cite[Thm.~X.17]{reed1975ii} and \cite[Ch.~VI-\S3]{kato2013perturbation}, taking into account the relative form boundedness of the previous proposition with $H_j =\alpha H_0 + V_j$. For the lower bound, the same relative form bound gives, for $\Psi\in\hfrak^+$,
\begin{equation}
    |\langle\Psi,V_j\Psi\rangle|\leq\frac{2\|f\|_{-1}}{\alpha}\langle\Psi,(\alpha H_0+\alpha)\Psi\rangle = 2\|f\|_{-1}\left(\langle\Psi,H_0\Psi\rangle+\|\Psi\|^2\right),
\end{equation}
whence, using $H_0\geq0$ and $\alpha>2\|f\|_{-1}$,
\begin{equation}
    \langle\Psi,H_j\Psi\rangle\geq\left(\alpha-2\|f\|_{-1}\right)\langle\Psi,H_0\Psi\rangle-2\|f\|_{-1}\|\Psi\|^2\geq-2\|f\|_{-1}\|\Psi\|^2 . \qedhere
\end{equation}
\end{proof}

\begin{remark}
Note that $H_j$ need not be positive: the operator $H_0$ has $0$ in its spectrum, the vacuum being annihilated by $\dOmega$ and $K$ vanishing on the ground state of the spin, so the perturbation $V_j$ may push the bottom of the spectrum below zero. What is used repeatedly below is the weaker statement~\eqref{eq:semibounded}, which guarantees that $H_j+\alpha$ is strictly positive and hence that $\|\cdot\|_{\pm,H_j}$ are genuine norms.
\end{remark}

To obtain the desired results we need to consider the interaction picture. This is obtained by a time-dependent unitary transformation generated by the free operator $\alpha H_0$. If $\Psi_j(t) = U_j(t)\Phi$ is the solution with initial value $\Phi\in\hfrak^+$ at time $t=0$ of the Schrödinger equation determined by the operator $H_j$, then $\hat{\Psi}_j(t) = \e^{\ii t\alpha H_{0}}U_j(t)\Phi$ is called the solution in the interaction picture and satisfies the Schrödinger equation

\begin{equation}
 \ii\frac{\mathrm{d}}{\mathrm{d}t}\hat\Psi_j(t)={\hat H}_j(t)\hat\Psi_j(t),\qquad j=1,2,
\end{equation}
in the sense of $\hfrak^-$, where ${\hat H}_j(t)$ is the so-called interaction-picture Hamiltonian,

\begin{equation}
 {\hat H}_j(t) = \e^{\ii t\alpha H_{0}}\,V_j\, \e^{-\ii t\alpha H_{0}}.
\end{equation}

Accordingly we write
\begin{equation}\label{eq:hatU}
    \hat U_j(t):=\e^{\ii t\alpha H_{0}}U_j(t)
\end{equation}
for the interaction-picture propagators, so that $\hat\Psi_j(t)=\hat U_j(t)\Phi$; more generally, if the initial time is $s\neq0$, we set $\hat U_j(t,s):=\e^{\ii t\alpha H_{0}}U_j(t-s)\e^{-\ii s\alpha H_{0}}$, so that $\hat U_j(t)=\hat U_j(t,0)$. Since $\e^{\ii t\alpha H_{0}}$ is unitary on $\hfrak$ one has
\begin{equation}\label{eq:same_error}
    \left\|\left[\hat U_2(t)-\hat U_1(t)\right]\Phi\right\| = \left\|\left[U_2(t)-U_1(t)\right]\Phi\right\|,
\end{equation}
so that estimating the error in the interaction picture is equivalent to estimating it in the Schr\"odinger picture. Moreover $\e^{\ii t\alpha H_{0}}$ commutes with $H_0$, hence acts isometrically for the norms $\|\cdot\|_{\pm,\alpha}$ and $\|\cdot\|_\pm$. Notice that $\e^{\ii t\alpha H_{0}}H_j \e^{-\ii t\alpha H_{0}} = \e^{\ii t\alpha H_{0}}\alpha H_0 \e^{-\ii t\alpha H_{0}} + \hat H_j(t) = \alpha H_0+ \hat H_j(t)$.

The purpose of this transformation is that, in the rotating frame, the time average of the difference $\hat{H}_1(t)-\hat{H}_2(t)$ is small. To see this we compute $\hat{H}_1(t)$ and $\hat{H}_2(t)$ explicitly.

Since $H_0$ is a sum of commuting terms, its propagator factors as
\begin{equation}
    \e^{-\ii t\alpha H_{0}}=\e^{-\ii t\alpha \omega_0 K}\otimes\e^{-\ii t\alpha\dOmega},
\end{equation}
with spin component
\begin{equation}\label{eq:spin_component}
\e^{-\ii t\alpha \omega_0 K}=\begin{pmatrix} \e^{-\ii t\alpha\omega_0}&0\\0&1\end{pmatrix}.
\end{equation}
A direct matrix computation gives $\e^{\ii t\alpha \omega_0 K}K\e^{-\ii t\alpha \omega_0 K}=K$ and, since $K$ acts diagonally in the spin basis,
\begin{equation}
    \e^{\ii t\alpha \omega_0 K}\sigma_+\e^{-\ii t\alpha \omega_0 K}=\e^{+\ii t\alpha\omega_0}\sigma_+,\qquad \e^{\ii t\alpha \omega_0 K}\sigma_-\e^{-\ii t\alpha \omega_0 K}=\e^{-\ii t\alpha\omega_0}\sigma_-,
\end{equation}
in summary, $\e^{\ii t\alpha \omega_0 K}\sigma_\pm\e^{-\ii t\alpha \omega_0 K}=\e^{\pm\ii t\alpha\omega_0}\sigma_\pm$.

The field component is handled by the following lemma.
\begin{lemma}\label{lem:field_component}
Let $f\in\sps^-$. Then $\e^{\ii t\alpha\dOmega}\a{f}\e^{-\ii t\alpha\dOmega}=\a{f_t}$ and $\e^{\ii t\alpha\dOmega}\adag{f}\e^{-\ii t\alpha\dOmega}=\adag{f_t}$, where $f_t(k):=\e^{\ii\alpha\omega(k)t}f(k)$.
\end{lemma}
\begin{proof}
Since $\e^{\ii t\alpha\dOmega}$ acts as multiplication on each $n$-particle sector,
\begin{equation}
(\e^{\ii\alpha\dOmega t}\Psi)^{(n)}(k_1,\dots,k_n)=\exp\left(\ii t\sum_{j=1}^n\alpha\omega(k_j)\right)\Psi^{(n)}(k_1,\dots,k_n),
\end{equation}
and the annihilation operator acts as
\begin{equation}
    (\a{f}\Psi)^{(n)}(k_1,\dots,k_n)=\sqrt{n+1}\int\overline{f(\kappa)}\Psi^{(n+1)}(k_1,\dots,k_{n},\kappa)\;\mathrm{d}\mu(\kappa),
\end{equation}
a direct computation gives, for every $n\geq 0$,
\begin{align}
   &\left(\a{f}\e^{-\ii\alpha\dOmega t}\Psi\right)^{(n)}(k_1,\dots,k_n)\nonumber\\
   &\quad=\sqrt{n+1}\int\overline{f(\kappa)}\exp\!\left(-\ii t\alpha\Bigl(\sum_{j=1}^n\omega(k_j)+\omega(\kappa)\Bigr)\right)\Psi^{(n+1)}(k_1,\dots,k_{n},\kappa)\;\mathrm{d}\mu(\kappa)\nonumber\\
   &\quad=\sqrt{n+1}\int\overline{f(\kappa)\e^{\ii\alpha\omega(\kappa)t}}\exp\!\left(-\ii \alpha t\sum_{j=1}^n\omega(k_j)\right)\Psi^{(n+1)}(k_1,\dots,k_{n},\kappa)\;\mathrm{d}\mu(\kappa),
\end{align}
whence
\begin{align}
     \left(\e^{+\ii\alpha\dOmega t}\a{f}\e^{-\ii\alpha\dOmega t}\Psi\right)^{(n)}(k_1,\dots,k_n)&=\sqrt{n+1}\int\overline{f(\kappa)\e^{\ii\alpha\omega(\kappa)t}}\Psi^{(n+1)}(k_1,\dots,k_{n},\kappa)\;\mathrm{d}\mu(\kappa)\nonumber\\
     &=(\a{f_t}\Psi)^{(n)}(k_1,\dots,k_n).
\end{align}
The creation operator identity follows by adjunction.
\end{proof}

\begin{proposition}\label{prop:expressions}
    The maps $\hat{H}_j(t)$, $j=1,2$, are given by
    \begin{align}
        \hat{H}_1(t)&=\sigma_+\otimes\left(\a{f_t^-}+\adag{f_t^+}\right)+\sigma_-\otimes\left(\a{f^+_t}+\adag{f_t^-}\right)\\
        \hat{H}_2(t)&=\sigma_+\otimes\a{f_t^-}+\sigma_-\otimes\adag{f_t^-},
    \end{align}
    where
    \begin{equation}
        f^\pm_t(k)=\e^{\ii t\alpha(\omega(k)\pm\omega_0)}f(k).
    \end{equation}
\end{proposition}

\begin{proof}
    By definition,
    \begin{align}
        \hat{H}_1(t)&=\e^{\ii t\alpha H_{0}}\left[(\sigma_++\sigma_-)\otimes(\a{f}+\adag{f})\right]\e^{-\ii t\alpha H_{0}},\\
        \hat{H}_2(t)&=\e^{\ii t\alpha H_{0}}\left[\sigma_+\otimes\a{f}+\sigma_-\otimes\adag{f}\right]\e^{-\ii t\alpha H_{0}}.
    \end{align}
Using $\e^{-\ii t\alpha H_0}=\e^{-\ii t\alpha \omega_0 K}\otimes\e^{-\ii t\alpha \dOmega}$, together with the identities established above for the spin component and with Lemma~\ref{lem:field_component} for the field component,
\begin{align}
      \e^{\ii t\alpha H_0}\sigma_\pm\otimes\a{f}\e^{-\ii t\alpha H_0}&=\sigma_\pm\otimes\a{\e^{\mp\ii t\alpha\omega_0}f_t},\\
    \e^{\ii t\alpha H_0}\sigma_\pm\otimes\adag{f}\e^{-\ii t\alpha H_0}&=\sigma_\pm\otimes\adag{\e^{\pm\ii t\alpha\omega_0}f_t},
\end{align}
the claim follows.
\end{proof}

\begin{remark}\label{rem:intuition}
From Proposition~\ref{prop:expressions}, the difference between the two interaction-picture Hamiltonians is
\begin{equation}\label{eq:difference}
    \hat{H}_1(t)-\hat{H}_2(t)=\sigma_+\otimes\adag{f_t^+}+\sigma_-\otimes\a{f_t^+},
\end{equation}
which consists precisely of the so-called \textit{counter-rotating terms}, oscillating at frequency $\alpha(\omega(k)+\omega_0)$. Note in particular that the $\Delta K\otimes\mathrm{I}$ terms would have cancelled exactly had they been present. The relative action $\hat{S}_{21}(t)$ therefore accumulates only these counter-rotating terms, and its smallness is a direct consequence of their rapid oscillation: since $\omega(k)+\omega_0\geq m+\omega_0>0$ for all $k\in X$, the frequency $\alpha(\omega(k)+\omega_0)$ diverges as $\alpha\to+\infty$ uniformly in $k$, so the integral of~\eqref{eq:difference} over any time interval decays as $O(\alpha^{-1})$. By contrast, the co-rotating terms in $\hat{H}_2(t)$ oscillate at frequency $|\omega(k)-\omega_0|$, which can vanish for $k$ near resonance; these are precisely the terms retained by the RWA.
\end{remark}

\subsection{Bounding the error of the RWA}\label{sec:rwa_bound_subsec}

Before we can apply the theorem we will need some preparatory results. As in the previous sections $H_1, V_1, \hat{H}_1$ will denote, respectively, the Hamiltonian, the interaction term and the interaction-picture Hamiltonian of the exact dynamics; $H_2, V_2, \hat{H}_2$ will denote the corresponding operators in the RWA approximation. Besides the norms of Lemma~\ref{lem:scale_norms} we shall also need the norms induced by the interacting Hamiltonians themselves,
\begin{equation}\label{eq:Hj_norm}
    \|\Psi\|_{+,H_j}:=\sqrt{\langle \Psi, (H_j +\alpha) \Psi\rangle},
\end{equation}
with $\|\cdot\|_{-,H_j}$ the corresponding dual norm. The following lemma establishes the relevant equivalences.

\begin{lemma}\label{lem:equivalences}
Let $\|\cdot\|_{+,H_j}$ be as in Eq.~\eqref{eq:Hj_norm}. Then
  \begin{enumerate}
  \item There exists $\alpha_0$ such that, for all $\alpha>\alpha_0$, $j=1,2$ and $\Psi\in\hfrak^+$,
    
\begin{equation}
\frac{1}{2} \|\Psi\|_{\pm,H_j} \leq \|\Psi\|_{\pm, \alpha} \leq 2 \|\Psi\|_{\pm,H_j},
\end{equation}

  \item There exist $\alpha_0$ and $C>0$ such that, for all $\alpha>\alpha_0$, $j=1,2$ and $\Psi\in\Dom(H_0)\cap\Dom(H_j)$,
    
\begin{equation}
 \| \alpha H_0 \Psi\| \leq C \|(H_j+\alpha)\Psi\|.
\end{equation}

  \end{enumerate}
\end{lemma}

\begin{proof}
To prove \textit{(1)} we use the relative form boundedness of $V_j$
\begin{align}
  |\langle\Psi, (\alpha H_0 +\alpha)\Psi \rangle| 
    &\leq  |\langle\Psi, (\alpha H_0+ \alpha + V_j) \Psi \rangle| +  |\langle\Psi, V_j \Psi \rangle|\\
    &\leq  |\langle\Psi, (H_j + \alpha )\Psi \rangle| + \frac{2\|f\|_{-1}}{\alpha} \langle \Psi, \alpha H_0 \Psi \rangle  +   \frac{2\|f\|_{-1}}{\alpha} \alpha\langle \Psi, \Psi \rangle. 
  \end{align}
From this inequality, assuming $\alpha$ large enough and using the relative form boundedness again, we obtain
\begin{equation}
 \left(1- \frac{2\|f\|_{-1}}{\alpha}\right)  |\langle\Psi, (\alpha H_0 +\alpha)\Psi \rangle|  \leq  |\langle\Psi, (H_j+\alpha) \Psi \rangle|  \leq  \left(1+ \frac{2\|f\|_{-1}}{\alpha}\right)|\langle\Psi, (\alpha H_0 +\alpha)\Psi \rangle|. 
\end{equation}  
The proof in the dual spaces follows by standard arguments in the scale of Hilbert spaces. 
To prove \textit{(2)} notice first that, by~\eqref{eq:semibounded}, $\|(H_j+\alpha)\Psi\|\geq (\alpha-2\|f\|_{-1})\|\Psi\|\geq\frac{\alpha}{2}\|\Psi\|$ whenever $\alpha\geq4\|f\|_{-1}$, and therefore $\|H_j\Psi\|\leq\|(H_j+\alpha)\Psi\|+\alpha\|\Psi\|\leq 3 \|(H_j+\alpha)\Psi\|$. Hence we have
\begin{equation}
\|\alpha H_0 \Psi\| \leq  \|H_j\Psi\| +  \|V_j \Psi\| \leq \|H_j\Psi\| + \frac{a}{\alpha}\|\alpha H_0 \Psi\| + b \|\Psi\|,
\end{equation}
from which it follows
\begin{equation}
\left(1-\frac{a}{\alpha}\right)\|\alpha H_0 \Psi\| \leq  \|H_j\Psi\| + b \|\Psi\| \leq 3\|(H_j+\alpha)\Psi\| + \frac{2b}{\alpha}\|(H_j+\alpha)\Psi\|.
\end{equation}
\end{proof}

What we shall do next is to prove that we can apply Theorem~\ref{thm:bound} to the interaction picture Hamiltonians. We will need several preparatory lemmas.

\begin{lemma}\label{lem:bound_potential}
Let $\Psi\in\hfrak^+$. For $j=1,2$, we have
\begin{equation}
\|V_j\Psi\|_- \leq 2 \|f\|_{-1}\|\Psi\|_+.
\end{equation}
\end{lemma}

\begin{proof}
It is enough to prove it for elements of the form $\Psi=u\otimes \psi \in \mathbb{C}^2\otimes \fock^+$. Using Proposition~\ref{prop:af} and Lemma~\ref{lem:scale_norms} we get
\begin{align}
  \|V_j \Psi\|_- 
    &\leq \|V_j \Psi\|_{\mathbb{C}^2\otimes\fock^-} \\
    &\leq \|u\|_{\mathbb{C}^2}\|\a{f}\psi\|_{\fock^-} + \|u\|_{\mathbb{C}^2}\|a^\dagger(f)\psi\|_{\fock^-} \\
    &\leq 2 \|u\|_{\mathbb{C}^2}\|\psi\|_{\fock^+}\|f\|_{-1} \\
    &\leq 2 \|u\otimes\psi\|_{+}\|f\|_{-1}. \qedhere
\end{align}
\end{proof}

\begin{lemma}\label{lem:commutator_estimate}
  Let $\Psi\in\hfrak^+$ and assume Hypothesis~\ref{hyp:f}. The commutator $[H_0,V_j]$, $j=1,2$, satisfies the estimate  
\begin{equation}
\|[H_0,V_j]\Psi\|_-\leq 2\left(\omega_0\|f\|_{-1} + \|\omega f\|_{-1}\right)\|\Psi\|_+.
\end{equation}
\end{lemma}

\begin{proof}
  This is a straightforward calculation using the canonical commutation relations, cf.~\cite{bratteli1997,reed1975ii}, in the form
\begin{equation}
[\dOmega, \a{f}] = -a(\omega f), \qquad [\dOmega, a^\dagger(f)] = a^\dagger(\omega f).
\end{equation}
  Taking into account the definition of the interaction terms $V_j$ of Eqs.~\eqref{eq:h} and~\eqref{eq:hrwa} and Proposition~\ref{prop:af}
  \begin{align}
  \|[H_0,V_j]\Psi\|_- 
    &\leq \|[H_0,V_j]\Psi\|_{\mathbb{C}^2\otimes\fock^-}\\
    &\leq 2\omega_0\|f\|_{-1}\|\Psi\|_{\mathbb{C}^2\otimes\fock^+} + 2 \|\omega f\|_{-1}\|\Psi\|_{\mathbb{C}^2\otimes\fock^+}\\
    &\leq 2\left(\omega_0\|f\|_{-1} + \|\omega f\|_{-1}\right)\|\Psi\|_+. \qedhere
  \end{align}
\end{proof}

\begin{lemma}\label{lem:potential_estimate}
Let $\Psi\in\Dom\left((H_0+1)^{3/2}\right)$. Then there exists $C>0$ independent of $\alpha$ such that
\begin{equation}
\|(\alpha H_0 +\alpha)^{1/2} V_j \Psi\| \leq \frac{C}{\alpha} \|(\alpha H_0 + \alpha)^{3/2}\Psi\| .
\end{equation}
\end{lemma}

\begin{proof}
First notice that $\|(\alpha H_0 +\alpha)^{1/2} V_j \Psi\| \leq  \| V_j(\alpha H_0 +\alpha)^{1/2} \Psi\| + \|[(\alpha H_0 +\alpha)^{1/2}, V_j] \Psi\|$. From the properties of the commutator,
\begin{equation}
\left[(\alpha H_0 +\alpha)^{1/2},V_j \right] = (\alpha H_0 +\alpha)^{-1/2}V_j(\alpha H_0 +\alpha) - V_j(\alpha H_0 +\alpha)^{1/2} + (\alpha H_0 +\alpha)^{-1/2}\left[\alpha H_0 + \alpha, V_j \right].
\end{equation}
Now we can bound the different terms.
First, using Lemma~\ref{lem:bound_potential}  we have
\begin{align}
  \|(\alpha H_0 +\alpha)^{-1/2}V_j(\alpha H_0 +\alpha)\Psi\| 
    &=\|V_j(\alpha H_0 +\alpha)\Psi\|_{-,\alpha} \\
    &=\frac{1}{\sqrt{\alpha}}\|V_j(\alpha H_0 +\alpha)\Psi\|_{-}\\
    &\leq \frac{2}{\sqrt{\alpha}}\|f\|_{-1}\|(\alpha H_0 +\alpha)\Psi\|_+\\
    &=\frac{2}{{\alpha}}\|f\|_{-1}\|(\alpha H_0 +\alpha)\Psi\|_{+,\alpha}\\
    &=\frac{2}{{\alpha}}\|f\|_{-1}\|(\alpha H_0 +\alpha)^{3/2}\Psi\|.
\end{align}
The next term can be bounded using the relative boundedness of the potential.
\begin{align}
  \|V_j(\alpha H_0 +\alpha)^{1/2}\Psi \| 
    &\leq \frac{a}{\alpha}\|\alpha H_0 (\alpha H_0 +\alpha)^{1/2} \Psi\| + b \|(\alpha H_0 +\alpha)^{1/2}\Psi\|\\
    &\leq  \frac{a}{\alpha}\|(\alpha H_0 +\alpha)^{3/2} \Psi\| + \frac{a}{\alpha}\|\alpha(\alpha H_0 +\alpha)^{1/2} \Psi\| + b \|(\alpha H_0 +\alpha)^{1/2}\Psi\|.
\end{align}
Noticing that $\alpha\|\Psi\|\leq\|(\alpha H_0 +\alpha)\Psi\|$ it follows that $\|\alpha(\alpha H_0 +\alpha)^{1/2} \Psi\|\leq \|(\alpha H_0 +\alpha)^{3/2} \Psi\|$, from which the latter two terms in the above inequality can be bounded to lead
\begin{equation}
\|V_j(\alpha H_0 +\alpha)^{1/2}\Psi \| \leq \frac{C'}{\alpha}\|(\alpha H_0 +\alpha)^{3/2}\Psi\|.
\end{equation}

The third term follows from the bound of the commutator in Lemma~\ref{lem:commutator_estimate}
\begin{align}
  \|(\alpha H_0 +\alpha)^{-1/2}\left[\alpha H_0 + \alpha, V_j \right]\Psi\|
    &= \|\left[\alpha H_0 + \alpha, V_j \right]\Psi\|_{-,\alpha}\\
    &= \frac{1}{\sqrt{\alpha}}\|\left[\alpha H_0 + \alpha, V_j \right]\Psi\|_{-}\\
    &=  {\sqrt{\alpha}}\|\left[H_0, V_j \right]\Psi\|_{-}\\
    &\leq C'' {\sqrt{\alpha}}\|\Psi\|_+\\
    &= C'' \|\Psi\|_{+,\alpha}\\
    &= C'' \|(\alpha H_0 + \alpha)^{1/2}\Psi\| \leq \frac{C''}{\alpha} \|(\alpha H_0 + \alpha)^{3/2}\Psi\|. \qedhere
\end{align}
\end{proof}

\begin{lemma}\label{lem:equivalent}
  For $j=1,2$ one has $\Dom\left((\alpha H_0 + \alpha)^{3/2}\right) = \Dom\left((H_j + \alpha)^{3/2}\right)$; moreover there exist $C$ and $\alpha_0$ such that, for all $\alpha>\alpha_0$ and all $\Psi$ in that common domain,
  \begin{align}
    \|(\alpha H_0 + \alpha)^{3/2}\Psi\| &\leq C \|(H_j + \alpha)^{3/2}\Psi\|\\
    \|(H_j + \alpha)^{3/2}\Psi\| &\leq C \|(\alpha H_0 + \alpha)^{3/2}\Psi\| .
  \end{align}
\end{lemma}

\begin{proof}
  Let $\Psi\in\Dom\left((H_j + \alpha)^{3/2}\right)$. Using \textit{(1)} of Lemma~\ref{lem:equivalences},
  \begin{align}
    \|(\alpha H_0 + \alpha)^{3/2}\Psi\| 
      &\leq 2 \|(H_j+\alpha)^{1/2}(\alpha H_0 + \alpha)\Psi\|\\
      &\leq 2 \|(H_j+\alpha)^{3/2}\Psi\| + 2\|(H_j+\alpha)^{1/2}V_j\Psi\|
  \end{align}

  To bound the second term we will use a similar argument as in the previous lemma using the properties of the commutator, from which  
\begin{equation}
\|(H_j+\alpha)^{1/2}V_j\Psi\| \leq \|V_j(H_j+\alpha)\Psi\|_{-,H_j} + \|[\alpha H_0,V_j]\Psi\|_{-,H_j} + 2\|V_j (H_j+\alpha)^{1/2}\Psi\|.
\end{equation}

The bounds on these three terms are obtained in a similar way to the corresponding bounds in the proof of Lemma~\ref{lem:potential_estimate}. We shall only provide the steps that differ in the corresponding proofs.
\begin{align}
\|V_j(H_j+\alpha)\Psi\|_{-,H_j}
  &\leq \frac{2}{\sqrt{\alpha}}\|V_j(H_j+\alpha)\Psi\|_-\\
  &\leq \frac{C'}{{\alpha}}\|(H_j+\alpha)\Psi\|_{+,\alpha}\\
  &\leq \frac{C''}{{\alpha}}\|(H_j+\alpha)\Psi\|_{+,H_j}\\
  &=\frac{C''}{{\alpha}}\|(H_j+\alpha)^{3/2}\Psi\|.
\end{align}
For the second term we have, using the equivalences of the norms and Lemma~\ref{lem:commutator_estimate},
\begin{align}
  \|[\alpha H_0,V_j]\Psi\|_{-,H_j}
    &\leq C\sqrt{\alpha}\|[H_0,V_j]\Psi\|_-\\
    &\leq C'\sqrt{\alpha}\|\Psi\|_+\\
    &\leq C''\|\Psi\|_{+,H_j} \leq \frac{C'''}{\alpha}\|(H_j+\alpha)^{3/2}\Psi\|.
\end{align}
The bound on the third term follows from the relative boundedness of the potential.
\begin{align}
  \|V_j (H_j+\alpha)^{1/2}\Psi\|
    &\leq \frac{a}{\alpha}\|\alpha H_0(H_j+\alpha)^{1/2}\Psi \| + b\|(H_j+\alpha)^{1/2}\Psi\|\\
    &\leq \frac{a}{\alpha}\|(H_j+\alpha)^{3/2}\Psi \| + \frac{a}{\alpha}\|V_j(H_j+\alpha)^{1/2}\Psi \| + (a+b)\|(H_j+\alpha)^{1/2}\Psi\|
\end{align}
For $\alpha>a$ we have
\begin{equation}
\|V_j (H_j+\alpha)^{1/2}\Psi\| \leq \frac{C}{\alpha}\|(H_j+\alpha)^{3/2}\Psi \|.
\end{equation}
  Combining the previous bounds proves the first inequality in the statement of the lemma.

  For the second inequality let $\Psi\in\Dom\left((\alpha H_0 + \alpha)^{3/2}\right)$. Using the equivalence of the norms we have  
\begin{equation}
\|(H_j+\alpha)^{3/2}\Psi\| \leq 2 \|(\alpha H_0+\alpha)^{3/2}\Psi\| + 2 \|(\alpha H_0+\alpha)^{1/2}V_j\Psi\|. 
\end{equation}
  Using Lemma~\ref{lem:potential_estimate} in the second term provides the result. 
\end{proof}

\begin{proposition}\label{prop:energy_bounds}
  Let $H_j$, $j=1,2$, be the self-adjoint, semibounded operators of Proposition~\ref{prop:sa_and_bound}, and let $\alpha>2\|f\|_{-1}$. Let $U_j(t)= \e^{-\ii tH_j}$ be, respectively, the strongly continuous one-parameter unitary groups defined by them. Then $U_j(t)\Dom\left((H_j+\alpha)^{r}\right)=\Dom\left((H_j+\alpha)^{r}\right)$ for every $r\geq0$, and
\begin{equation}\label{eq:conservation}
\left\|(H_j+\alpha)^{r}U_j(t)\Psi\right\| = \left\|(H_j+\alpha)^{r}\Psi\right\|\qquad\text{for all }\Psi\in\Dom\left((H_j+\alpha)^{r}\right) .
\end{equation}
In particular, for $r=1/2$ and $r=3/2$,
\begin{align}
\|U_j(t)\Psi\|_{+,H_j} &= \|\Psi\|_{+,H_j} &&\text{for all }\Psi\in\hfrak^+=\Dom\left((H_j+\alpha)^{1/2}\right),\\
\left\|(H_j+\alpha)^{3/2}U_j(t)\Psi\right\| &= \left\|(H_j+\alpha)^{3/2}\Psi\right\| &&\text{for all }\Psi\in\Dom\left((H_j+\alpha)^{3/2}\right) .
\end{align}
\end{proposition}

\begin{proof}
  By the spectral theorem $U_j(t)$ is unitary and commutes with every Borel function of $H_j$, in particular with $(H_j+\alpha)^{r}$, and preserves its domain; note that $H_j+\alpha>0$ by~\eqref{eq:semibounded}, so that these powers are well defined. Hence, for $\Psi\in\Dom\left((H_j+\alpha)^{r}\right)$,
\begin{equation}
\left\|(H_j+\alpha)^{r}U_j(t)\Psi\right\| = \left\|U_j(t)(H_j+\alpha)^{r}\Psi\right\| = \left\|(H_j+\alpha)^{r}\Psi\right\| ,
\end{equation}
which is~\eqref{eq:conservation}; the two displayed instances follow by the definition~\eqref{eq:Hj_norm} of $\|\cdot\|_{+,H_j}$. 
\end{proof}

Finally, we check that the solutions in the interaction picture inherit the regularity of $\Psi_j(t)$. The operator $\e^{\ii t\alpha H_0}$ relating them commutes with $H_0$, hence with $(\alpha H_0+\alpha)^{1/2}$; consequently it restricts to a strongly continuous one-parameter unitary group on $\hfrak^+$ and extends to one on $\hfrak^-$, acting isometrically for both norms $\|\cdot\|_{\pm,\alpha}$. 
Since $\alpha H_0\in\mathcal{B}(\hfrak^+,\hfrak^-)$, for every $\Phi\in\hfrak^+$ the map $t\mapsto \e^{\ii t\alpha H_0}\Phi$ is differentiable in $\hfrak^-$, with derivative $\ii\alpha H_0\e^{\ii t\alpha H_0}\Phi$. Combining this with the product rule and with the Schr\"odinger equation for $\Psi_j(t)$, which holds in $\hfrak^-$, shows that $\hat{\Psi}_j(t)$ satisfies Hypothesis~\ref{hyp:2}(i).

We can now identify the subspace $\mathcal{D}$ postulated in Hypothesis~\ref{hyp:2}(ii). By Lemma~\ref{lem:equivalent} the domains of $(\alpha H_0+\alpha)^{3/2}$ and of $(H_j+\alpha)^{3/2}$ coincide, for $j=1,2$; moreover $\alpha$ is a positive multiplicative constant, so that this common domain agrees with that of $(H_0+1)^{3/2}$ and is independent both of $j$ and of $\alpha$. We therefore set
\begin{equation}\label{eq:Ddef}
    \mathcal{D}:=\Dom\left((H_0+1)^{3/2}\right)=\Dom\left((\alpha H_0 + \alpha)^{3/2}\right)=\Dom\left((H_j + \alpha)^{3/2}\right),\qquad j=1,2 .
\end{equation}
\emph{This is the subspace $\mathcal{D}$ of Hypothesis~\ref{hyp:2}(ii)}: Proposition~\ref{prop:uniform_energy_bound} below shows that it satisfies parts {(a)}--{(c)} of that hypothesis, and it is dense in $\hfrak$, being the domain of a power of a self-adjoint operator. Since we also have $\Dom((H_j + \alpha)^{1/2}) = \Dom((\alpha H_0 + \alpha)^{1/2})=\hfrak^+$ and the one-parameter unitary groups preserve $\hfrak^+$, we obtain 
\begin{equation}
H_jU_k(t)\Psi \in \Dom\left((H_j+\alpha)^{1/2}\right)=\hfrak^+;\quad k=1,2;\; j=1,2.
\end{equation}

\begin{proposition}\label{prop:uniform_energy_bound}
  The subspace $\mathcal{D}$ of~\eqref{eq:Ddef} and the interaction-picture Hamiltonians $\hat{H}_j(t)$, $j=1,2$, satisfy Hypothesis~\ref{hyp:2}(ii). In particular, for every $\Psi\in\mathcal{D}$,
  \begin{equation}\label{eq:MPsi}
      \|\hat{H}_k(\tau)\hat{U}_j(t,s)\Psi\|_+\leq M_\Psi:=C\,\|( H_0+1)^{3/2}\Psi\|\qquad\text{for all }j,k=1,2\text{ and all }\tau,t,s,
  \end{equation}
  with $C>0$ independent of $\Psi$, of $\tau,t,s$ and of $\alpha$.
\end{proposition}

\begin{proof}
  Property {(a)} holds because $\hat U_j(t,s)=\e^{\ii t\alpha H_0}U_j(t-s)\e^{-\ii s\alpha H_0}$, and all three factors preserve $\mathcal{D}$: the group $\e^{\ii\sigma\alpha H_0}$ commutes with $H_0$, hence preserves $\Dom((H_0+1)^{3/2})$, while $U_j(\sigma)$ commutes with $H_j$ and therefore preserves $\Dom((H_j+\alpha)^{3/2})$, which coincides with $\mathcal{D}$ by Lemma~\ref{lem:equivalent}. Property {(b)} was established in the discussion preceding this proposition. Density of $\mathcal{D}$ in $\hfrak$ is immediate, $\mathcal{D}$ being the domain of a power of a self-adjoint operator. It remains to prove~\eqref{eq:MPsi}, which is property {(c)}. Set $\chi:=\e^{-\ii s\alpha H_0}\Psi$, so that $\|(\alpha H_0+\alpha)^{3/2}\chi\|=\|(\alpha H_0+\alpha)^{3/2}\Psi\|$, and observe that
  \begin{equation}
      \hat{H}_k(\tau)\hat{U}_j(t,s)\Psi=\e^{\ii\tau\alpha H_0}V_k\,\e^{\ii(t-\tau)\alpha H_0}U_j(t-s)\chi ,
  \end{equation}
  the two free exponentials being unitary and commuting with every power of $H_0+1$. Hence
  \begin{align}
    \|\hat{H}_k(\tau)\hat{U}_j(t,s)\Psi\|_+ 
      &= \|(H_0+1)^{1/2}V_k\,\e^{\ii(t-\tau)\alpha H_0}U_j(t-s)\chi\| \\
      &= \frac{1}{\sqrt{\alpha}}\|(\alpha H_0+\alpha)^{1/2}V_k\,\e^{\ii(t-\tau)\alpha H_0}U_j(t-s)\chi\| \\
      \text{By Lemma~\ref{lem:potential_estimate}\quad} &\leq \frac{C}{{\alpha^{3/2}}}\|(\alpha H_0+\alpha)^{3/2}\e^{\ii(t-\tau)\alpha H_0}U_j(t-s)\chi\| \\
      &= \frac{C}{{\alpha^{3/2}}}\|(\alpha H_0+\alpha)^{3/2}{U}_j(t-s)\chi\| \\
      \text{By Lemma~\ref{lem:equivalent}\quad} &\leq \frac{C'}{{\alpha^{3/2}}}\|(H_j+\alpha)^{3/2}{U}_j(t-s)\chi\|\\
      \text{By Proposition~\ref{prop:energy_bounds}\quad} &\leq \frac{C''}{{\alpha^{3/2}}}\|(H_j+\alpha)^{3/2}\chi\|\\
      \text{By Lemma~\ref{lem:equivalent}\quad} &\leq \frac{C'''}{{\alpha^{3/2}}}\|(\alpha H_0+\alpha)^{3/2}\chi\|\\
      &= C'''\|( H_0+1)^{3/2}\Psi\| .
  \end{align}
Since $C'''$ does not depend on $\tau$, $t$, $s$, $j$, $k$ or $\alpha$, this is~\eqref{eq:MPsi}.

  It remains to check the continuity required by {(c)}. By Lemma~\ref{lem:potential_estimate} the operator $V_k(H_0+1)^{-3/2}$ is bounded from $\hfrak$ to $\hfrak^+$, and $\e^{\ii\tau\alpha H_0}$ acts isometrically on $\hfrak^+$; in view of the identity displayed above it therefore suffices that
  \begin{equation}
      (\tau,t,s)\longmapsto (H_0+1)^{3/2}\,\e^{\ii(t-\tau)\alpha H_0}U_j(t-s)\e^{-\ii s\alpha H_0}\Psi\in\hfrak
  \end{equation}
  be continuous. The two exponentials commute with $(H_0+1)^{3/2}$ and are strongly continuous, while
  \begin{equation}
      (H_0+1)^{3/2}U_j(\sigma)=\left[(H_0+1)^{3/2}(H_j+\alpha)^{-3/2}\right]U_j(\sigma)\,(H_j+\alpha)^{3/2}
  \end{equation}
  on $\mathcal{D}$, the bracket being bounded by Lemma~\ref{lem:equivalent} and $U_j$ strongly continuous. Continuity follows.
\end{proof}

We now apply Theorem~\ref{thm:bound}. By Remark~\ref{rem:intuition}, $\hat{H}_1(t)-\hat{H}_2(t)$ consists entirely of counter-rotating terms oscillating at frequency $\alpha(\omega(k)+\omega_0)$; the relative action therefore accumulates only these. By Proposition~\ref{prop:expressions},
\begin{align}
    \hat{S}_{21}(t)&=-\int_0^t\left[\sigma_+\otimes\adag{f_s^+}+\sigma_-\otimes\a{f_s^+}\right]\,\mathrm{d}s\nonumber\\
    &=-\left[\sigma_+\otimes\adag{F_t^+}+\sigma_-\otimes\a{F_t^+}\right],
\end{align}
the integral converging in $\mathcal{B}(\hfrak^+,\hfrak^-)$, where
\begin{align}
    F_t^+(k)&=\int_0^t f^+_s(k)\;\mathrm{d}s\nonumber\\
    &=f(k)\int_0^t \e^{\ii s\alpha(\omega(k)+\omega_0)}\;\mathrm{d}s\nonumber\\
    &=\frac{\e^{\ii\alpha t(\omega(k)+\omega_0)}-1}{\ii\alpha(\omega(k)+\omega_0)}f(k).
\end{align}
By Proposition~\ref{prop:af}, the operator norms of $\a{F^+_t},\adag{F^+_t}$ as maps from $\fock^+$ to $\fock^-$ satisfy
\begin{align}
    \left\|a^{\sharp}(F_t^+)\right\|_{\mathcal{B}(\fock^+,\fock^-)}^2&\leq\left\|F^+_t\right\|_{-1}^2\nonumber\\
    &=\int\frac{|F^+_t(k)|^2}{\omega(k)}\,\mathrm{d}\mu\nonumber\\
    &\leq\frac{4}{\alpha^2}\int\frac{|f(k)|^2}{\omega(k)\left[\omega(k)+\omega_0\right]^2}\,\mathrm{d}\mu\nonumber\\
    &\leq\frac{4}{\alpha^2(m+\omega_0)^2}\|f\|^2_{-1},
\end{align}
where $a^{\sharp}$ stands for either $a{}$ or $a^\dagger{}$.
Together with Lemma~\ref{lem:scale_norms}, which gives $\|\Psi\|_{+}\geq\|\Psi\|_{\mathbb{C}^2\otimes\fock^+}$ and 
\begin{equation}
\|\hat{S}_{21}(t)\Psi\|_{-}\leq\|\hat{S}_{21}(t)\Psi\|_{\mathbb{C}^2\otimes\fock^-},
\end{equation}
the $\mathcal{B}(\hfrak^+,\hfrak^-)$ norm of $\hat{S}_{21}(t)$ satisfies
\begin{align}
    \|\hat S_{21}(t)\|_{\mathcal{B}(\hfrak^+,\hfrak^-)}&=\sup_{0\neq\Psi\in\hfrak^+}\frac{\|\hat{S}_{21}(t)\Psi\|_{-}}{\|\Psi\|_{+}}\nonumber\\
    &\leq\sup_{0\neq\Psi\in\hfrak^+}\frac{\|\hat{S}_{21}(t)\Psi\|_{\mathbb{C}^2\otimes\fock^-}}{\|\Psi\|_{\mathbb{C}^2\otimes\fock^+}}\nonumber\\
    &=\left\|\hat{S}_{21}(t)\right\|_{\mathcal{B}(\mathbb{C}^2\otimes\fock^+,\mathbb{C}^2\otimes\fock^-)}\nonumber\\
    &=\left\|\sigma_+\otimes\adag{F_t^+}+\sigma_-\otimes\a{F_t^+}\right\|_{\mathcal{B}(\mathbb{C}^2\otimes\fock^+,\mathbb{C}^2\otimes\fock^-)}\nonumber\\
    &\leq\frac{4}{\alpha(m+\omega_0)}\|f\|_{-1}.\label{eq:S21bound}
\end{align}
Applying Theorem~\ref{thm:bound} to the interaction-picture forms now gives, for every $\Phi$ in the subspace $\mathcal{D}$ of~\eqref{eq:Ddef},
\begin{align}\label{eq:rwa_bound}
          \frac{1}{2} \left\|\left[\hat U_2(t)-\hat U_1(t)\right]\Phi\right\|^2\leq& \;\frac{4}{\alpha(m+\omega_0)}\|f\|_{-1}\bigg(\|\hat U_2(t)\Phi\|_+\|\hat U_1(t)\Phi\|_+\nonumber\\
           &+\int_0^t\bigg(\|\hat U_2(\tau)\Phi\|_+\|\hat{H}_1(\tau)\hat U_1(\tau)\Phi\|_+\nonumber\\
          &\qquad\qquad\qquad\qquad+\|\hat U_1(\tau)\Phi\|_+\|\hat{H}_2(\tau)\hat U_2(\tau)\Phi\|_+\bigg)\;\mathrm{d}\tau\biggr),
        \end{align}
whose left-hand side coincides with the Schr\"odinger-picture error by~\eqref{eq:same_error}.

It remains to make the state-dependent factors in~\eqref{eq:rwa_bound} explicit. This is the content of the following lemma, whose point is that the resulting constant does not depend on $\alpha$.

\begin{lemma}\label{lem:propagator_plus}
    There is $\alpha_0>0$ such that, for all $\alpha>\alpha_0$, all $t,s\in\mathbb{R}$, $j=1,2$, and all $\Phi\in\hfrak^+$,
    \begin{equation}\label{eq:propagator_plus}
        \left\|\hat U_j(t,s)\Phi\right\|_+\leq 4\,\|\Phi\|_+ .
    \end{equation}
\end{lemma}
\begin{proof}
    Set $\chi:=\e^{-\ii s\alpha H_0}\Phi$. Both free exponentials in $\hat U_j(t,s)=\e^{\ii t\alpha H_0}U_j(t-s)\e^{-\ii s\alpha H_0}$ are $\|\cdot\|_{+,\alpha}$-isometries, so that $\|\hat U_j(t,s)\Phi\|_{+,\alpha}=\|U_j(t-s)\chi\|_{+,\alpha}$ and $\|\chi\|_{+,\alpha}=\|\Phi\|_{+,\alpha}$. Two applications of Lemma~\ref{lem:equivalences}\textit{(1)}, one at each end, together with the conservation law of Proposition~\ref{prop:energy_bounds}, give
    \begin{equation}
        \|U_j(t-s)\chi\|_{+,\alpha}\leq 2\|U_j(t-s)\chi\|_{+,H_j}=2\|\chi\|_{+,H_j}\leq 4\|\chi\|_{+,\alpha}=4\|\Phi\|_{+,\alpha}.
    \end{equation}
    Dividing by $\sqrt\alpha$ and using Lemma~\ref{lem:scale_norms}\textit{(1)} yields~\eqref{eq:propagator_plus}; note that the factor $\sqrt\alpha$ cancels between the two sides.
\end{proof}

Combining~\eqref{eq:rwa_bound} with Lemma~\ref{lem:propagator_plus} and Proposition~\ref{prop:uniform_energy_bound} provides the final result.

\begin{theorem}\label{thm:explicit}
    Let $U(t)$ and $U_{\mathrm{RWA}}(t)$ be, respectively, the unitary groups generated by the Hamiltonian of Eq.~\eqref{eq:h} and by its rotating-wave approximation, Eq.~\eqref{eq:hrwa}, and let Hypothesis~\ref{hyp:f} hold. There exist $\alpha_0>0$ and $C>0$, the latter independent of $\alpha$, of $t$ and of $\Phi$, such that for all $\alpha>\alpha_0$, all $t\geq 0$ and all $\Phi\in\Dom((H_0 + 1)^{3/2})$,
    \begin{equation}\label{eq:explicit_bound}
        \left\|\left[U_{\mathrm{RWA}}(t)-U(t)\right]\Phi\right\|^2\leq \frac{C\,\|f\|_{-1}}{\alpha(m+\omega_0)}\left(\|\Phi\|_+^2 + t\, \|\Phi\|_+\left\|(H_0+1)\Phi\right\|_{+}\right).
    \end{equation}
\end{theorem}
\begin{proof}
    By~\eqref{eq:same_error} the left-hand side of~\eqref{eq:explicit_bound} equals $\|[\hat U_2(t)-\hat U_1(t)]\Phi\|^2$, so it suffices to estimate the right-hand side of~\eqref{eq:rwa_bound}. By Lemma~\ref{lem:propagator_plus} the boundary term is bounded by $16\|\Phi\|_+^2$. For the integral term, Lemma~\ref{lem:propagator_plus} and Proposition~\ref{prop:uniform_energy_bound} bound each of the two products in the integrand by $4C_0\|\Phi\|_+\|(H_0+1)^{3/2}\Phi\|$, with $C_0$ the constant of~\eqref{eq:MPsi}, uniformly in $\tau\in[0,t]$; integrating gives $8C_0\,t\,\|\Phi\|_+\|(H_0+1)^{3/2}\Phi\|$. Since $\|(H_0+1)^{3/2}\Phi\|=\|(H_0+1)\Phi\|_+$, multiplying~\eqref{eq:rwa_bound} by $2$ yields~\eqref{eq:explicit_bound} with $C=8\max\{16,8C_0\}$. None of the constants involved depends on $\alpha$, $t$ or $\Phi$.
\end{proof}

In particular, for any fixed $t$ and $\Phi$, the right-hand side vanishes as $\alpha\to\infty$, confirming that the RWA becomes exact in the limit of large frequency.

Theorem~\ref{thm:explicit} estimates the error in the norm of $\hfrak$, at the price of one half power of $\alpha$: the left-hand side of~\eqref{eq:explicit_bound} is quadratic in the error while the right-hand side is linear in the relative action. Corollary~\ref{cor:equality}, by contrast, is linear in the error on both sides, and the same ingredients then yield the full power $\alpha^{-1}$. The price is paid elsewhere, namely in the topology in which the error is measured. Indeed, the last term of Corollary~\ref{cor:equality} carries the test vector in the position $\hat H_1(\tau)\hat U_1(\tau,t)\Psi$, and Proposition~\ref{prop:uniform_energy_bound} controls that quantity through $\|(H_0+1)^{3/2}\Psi\|$ rather than through $\|\Psi\|_+$; the test vectors must therefore be normalised in $\mathcal{D}$ rather than in $\hfrak^+$. Writing
\begin{equation}\label{eq:Dnorm}
    \|\Psi\|_{\mathcal{D}}:=\left\|(H_0+1)^{3/2}\Psi\right\|,\qquad\Psi\in\mathcal{D},
\end{equation}
and using that $(H_0+1)^{3/2}$ maps $\mathcal{D}$ onto $\hfrak$, the norm of a vector $\chi\in\hfrak$ as a functional on $\left(\mathcal{D},\|\cdot\|_{\mathcal{D}}\right)$ is
\begin{equation}\label{eq:Ddual}
    \sup\left\{\left|\braket{\Psi,\chi}\right|\;:\;\Psi\in\mathcal{D},\;\|\Psi\|_{\mathcal{D}}\leq1\right\}=\left\|(H_0+1)^{-3/2}\chi\right\| ,
\end{equation}
which, since $H_0\geq0$, is weaker than $\|\chi\|_-=\|(H_0+1)^{-1/2}\chi\|$. It is in this topology that the rate $\alpha^{-1}$ is available.

\begin{corollary}\label{cor:minus_bound}
    Under the hypotheses of Theorem~\ref{thm:explicit} there exists $C'>0$, independent of $\alpha$, of $t$ and of $\Phi$, such that
    \begin{equation}\label{eq:minus_bound}
        \left\|(H_0+1)^{-3/2}\left[U_{\mathrm{RWA}}(t)-U(t)\right]\Phi\right\|\leq \frac{C'\,\|f\|_{-1}}{\alpha(m+\omega_0)}\left(\|\Phi\|_+ + t\,\left\|(H_0+1)\Phi\right\|_{+}\right).
    \end{equation}
\end{corollary}
\begin{proof}
    Let $\Psi\in\mathcal{D}$ with $\|\Psi\|_{\mathcal{D}}\leq1$; since $H_0\geq0$ one also has $\|\Psi\|_+=\|(H_0+1)^{1/2}\Psi\|\leq\|\Psi\|_{\mathcal{D}}\leq1$. Apply Corollary~\ref{cor:equality} to the interaction-picture forms with $s=0$, and estimate the three resulting terms by Lemma~\ref{lem:action}\textit{(iv)} together with the bound~\eqref{eq:S21bound} on $\|\hat S_{21}\|_{+,-}$, which we abbreviate as $S:=4\|f\|_{-1}/\left[\alpha(m+\omega_0)\right]$. By Lemma~\ref{lem:propagator_plus} the boundary term is bounded by $S\|\Psi\|_+\|\hat U_2(t)\Phi\|_+\leq4S\|\Phi\|_+$. For the two terms under the integral, Lemma~\ref{lem:propagator_plus} and Proposition~\ref{prop:uniform_energy_bound} give, uniformly in $\tau\in[0,t]$,
    \begin{align}
        \left|s_{21,\tau,0}\left(\hat U_1(\tau,t)\Psi,\hat H_2(\tau)\hat U_2(\tau)\Phi\right)\right|&\leq S\left\|\hat U_1(\tau,t)\Psi\right\|_+\left\|\hat H_2(\tau)\hat U_2(\tau)\Phi\right\|_+\leq 4C_0\,S\left\|(H_0+1)^{3/2}\Phi\right\| ,\\
        \left|s_{21,\tau,0}\left(\hat H_1(\tau)\hat U_1(\tau,t)\Psi,\hat U_2(\tau)\Phi\right)\right|&\leq S\left\|\hat H_1(\tau)\hat U_1(\tau,t)\Psi\right\|_+\left\|\hat U_2(\tau)\Phi\right\|_+\leq 4C_0\,S\left\|\Phi\right\|_+ ,
    \end{align}
    with $C_0$ the constant of~\eqref{eq:MPsi}; in the second estimate Proposition~\ref{prop:uniform_energy_bound} is applied to $\Psi$, whose $\mathcal{D}$-norm is at most $1$. This is the step that forces the normalisation of $\Psi$ in $\mathcal{D}$: the quantity $\|\hat H_1(\tau)\hat U_1(\tau,t)\Psi\|_+$ admits no bound in terms of $\|\Psi\|_+$ alone. Integrating over $[0,t]$ and using $\|\Phi\|_+\leq\|(H_0+1)^{3/2}\Phi\|=\|(H_0+1)\Phi\|_+$, we obtain
    \begin{equation}
        \left|\braket{\Psi,\left[\hat U_2(t)-\hat U_1(t)\right]\Phi}\right|\leq\frac{C'\|f\|_{-1}}{\alpha(m+\omega_0)}\left(\|\Phi\|_++t\left\|(H_0+1)\Phi\right\|_+\right)
    \end{equation}
    with $C'=16\max\{1,2C_0\}$. Taking the supremum over all $\Psi$ with $\|\Psi\|_D\leq 1$ gives, by~\eqref{eq:Ddual}, the quantity $\|(H_0+1)^{-3/2}[\hat U_2(t)-\hat U_1(t)]\Phi\|$. Since $\hat U_j(t)=\e^{\ii t\alpha H_0}U_j(t)$ by~\eqref{eq:hatU} and $\e^{\ii t\alpha H_0}$ is unitary and commutes with $(H_0+1)^{-3/2}$, this coincides with the corresponding Schr\"odinger-picture quantity, which is the left-hand side of~\eqref{eq:minus_bound}.
\end{proof}

\section{Discussion}\label{sec:discussion}

The bounds in Theorem~\ref{thm:explicit} and Corollary~\ref{cor:minus_bound} constitute a rigorous, quantitative and fully explicit justification of the rotating-wave approximation for spin--boson models with a structured boson field. We close by discussing their scope and related questions that remain open.

For a fixed initial state $\Phi\in\Dom((H_0+1)^{3/2})$ and a fixed time $t$, the right-hand side of~\eqref{eq:explicit_bound} decays as $O(\alpha^{-1})$, so that the RWA error itself vanishes as $O(\alpha^{-1/2})$ in the limit of large frequency. Every constant entering the estimate depends only on the parameters $\omega_0$, $m$, $\|f\|_{-1}$ and $\|f\|_{+1}$ of the model, and none of them on $\alpha$; the subspace $\Dom((H_0+1)^{3/2})$ on which the estimate is stated is likewise $\alpha$-independent and dense in $\hfrak$. This uniformity is not automatic, and is the reason for the particular choice of scales made in Section~\ref{sec:rwa}: the norms are built from the $\alpha$-independent operator $H_0$ rather than from the free Hamiltonian $\alpha H_0+\Delta K\otimes\mathrm{I}$, the parameter $\alpha$ being reinstated through the family $\|\cdot\|_{\pm,\alpha}$ of Eq.~\eqref{eq:+_norm} only where it is needed. Had the scales been defined through the free Hamiltonian, the norms themselves would have grown like $\sqrt\alpha$ and would have obscured the very decay one is trying to exhibit.

In contrast with the operator-norm bounds of \cite{burgarth2022one,dey2025floquet}, and like those of \cite{burgarth2024taming,richter2026dicke,burgarth2026floquet}, the estimate is state-dependent: the initial datum is required to lie one further level up in the scale, and the bound is proportional to $\|(H_0+1)\Phi\|_+$. This is unavoidable in the present framework, where the generators are only assumed to map $\hfrak^+$ into $\hfrak^-$; it is also, arguably, the more informative statement in the applications, since it quantifies how the quality of the approximation degrades for states of high energy. As for the dependence on time, the bound grows linearly in $t$ and carries no exponential factor. This is a consequence of the conservation law of Proposition~\ref{prop:energy_bounds}, which controls the propagated $\hfrak^+$-norm uniformly in time and is available because the Hamiltonians $H$ and $H_{\rm RWA}$ are time-independent, the time dependence being introduced only by the passage to the interaction picture.

The exponent $-1/2$ deserves a comment, since the bounds available in the monochromatic case~\cite{burgarth2024taming,richter2026dicke} are of order $\alpha^{-1}$. The difference is in the shape of the underlying identity: the left-hand side of the identity underlying Theorem~\ref{thm:bound} is \emph{quadratic} in the error, while the right-hand side is \emph{linear} in the relative action $S_{21}$, and the passage from one to the other necessarily halves the exponent. Working at the operator level on a common invariant domain, as in \cite{burgarth2022one,burgarth2024taming}, one disposes instead of an identity for the difference of the propagators itself, of the form~\eqref{eq:bound2}, linear on both sides, and this price is not paid. We note that a rate $O(\alpha^{-1})$ \emph{is} available in the present framework if one is content with measuring the error in the dual of $\mathcal{D}$, that is, after application of $(H_0+1)^{-3/2}$: this is the content of Corollary~\ref{cor:minus_bound}, whose starting point, Corollary~\ref{cor:equality}, is linear in the error on both sides. The trade-off is intrinsic to that identity, which pairs the error against a test vector to which the generators are then applied; the test vector must consequently be taken one full level higher in the scale than the $\hfrak^+$ that would produce the $\hfrak^-$ norm, and whether the intermediate topologies can be reached is open as well. Whether the exponent can be improved to $1$ in the Hilbert space norm, for instance by iterating the identity of Proposition~\ref{prop:equality} or by exploiting a cancellation between its boundary and integral terms, is an interesting open question.

A second limitation concerns the regularity of the form factor. Theorem~\ref{thm:explicit} requires Hypothesis~\ref{hyp:f}, that is, $\|\omega^{1/2}f\|<\infty$, which by Remark~\ref{rem:fplus} is, in general, strictly stronger than the condition $\|f\|_{-1}<\infty$ under which the model itself is well defined. The restriction is not an artefact of our estimates. Its source is the asymmetry between annihilation and creation operators: the commutator identity $[\dOmega,\adag{f}]=\adag{\omega f}$ weighs the form factor by an additional power of $\omega$, and no amount of regularity of the state can compensate for this. Consequently, the genuinely singular regime studied in \cite{dammoller,lonigro2022generalized,lonigro2023selfadjoint,Lill2025}, in which $f\notin L^2_\mu(X)$, lies outside the scope of the present analysis. It remains an open problem whether different estimates can extend the present analysis to less regular form factors. What the form-theoretic approach does buy is the treatment of a \emph{structured}, non-monochromatic field. In the monochromatic setting ($\mu=\delta_{k_0}$) the present bound reduces to a form-level counterpart of the estimates in \cite{burgarth2024taming,richter2026dicke}, which were obtained by operator-level methods on a common invariant domain.

Turning to the abstract framework, Hypothesis~\ref{hyp:1} and Hypothesis~\ref{hyp:2}(i) are mild and hold in wide generality, as Examples~\ref{ex:bilinear} and~\ref{ex:interaction} illustrate. Hypothesis~\ref{hyp:2}(ii) is the substantive one: as shown in Remark~\ref{rem:ii_needed}, it does not follow from the boundedness of the forms, not even for bounded perturbations, and verifying it requires quantitative control of the perturbation one level up in the scale. In the case at hand this is supplied by the relative operator boundedness of Proposition~\ref{prop:rel_bound} together with the commutator estimate of Lemma~\ref{lem:commutator_estimate}. Isolating these two ingredients as an abstract criterion, of which Section~\ref{sec:rwa} would then be a verification, would be a natural way of extending the reach of the framework.

Finally, the integration-by-parts technique used here produces a first-order effective Hamiltonian, namely $H_{\rm RWA}$. In the monochromatic case an iterated version of the same technique generates the full Floquet--Magnus hierarchy of effective Hamiltonians, with explicit error bounds at each order~\cite{dey2025floquet,burgarth2026floquet}. Whether the form-level identity of Proposition~\ref{prop:equality} can be iterated so as to produce higher-order effective Hamiltonians for structured fields is a natural question, and one that would require a form-theoretic counterpart of the invariant domains used at the operator level. A further direction concerns time-dependent generators: the uniform control of the propagated norm used here rests on energy conservation, and extending Theorem~\ref{thm:explicit} to genuinely time-dependent Hamiltonians would call for a Gr\"onwall-type substitute, with the attendant exponential growth in time.

\subsection*{Acknowledgments}

DL is grateful to Daniel Burgarth, Robin Hillier, and Leonhard Richter for stimulating discussions. JMPP acknowledges support from Agencia Estatal de Investigaci\'on, Project PID2024-160539NB-I00 and Comunidad de Madrid, Consejer\'ia de Educaci\'on, Ciencia y Universidades, Project QUITEMAD TEC-2024/COM-84.

\AtNextBibliography{\small}	\DeclareFieldFormat{pages}{#1}\sloppy 
	\printbibliography

\end{document}